\documentclass[11pt,letterpaper]{article}
\usepackage[utf8]{inputenc}
\usepackage{amsmath}
\usepackage{amsfonts}
\usepackage{amssymb}
\usepackage{amsthm}
\usepackage{array}
\newcolumntype{?}{!{\color{lightgray}\vrule width 0.01pt}}
\usepackage[left=2.6cm,right=2.6cm,top=2.6cm,bottom=2.6cm]{geometry}
\usepackage{enumitem}
\usepackage{gensymb}
\usepackage{graphicx}
\usepackage{verbatim}
\usepackage{dsfont}
\usepackage{bbm}
\usepackage{tikz}
\usepackage{algpseudocode}
\newcommand{\mb}{\mathbb}
\newcommand{\mc}{\mathcal}
\newcommand{\e}{\epsilon}

\allowdisplaybreaks
\usetikzlibrary{positioning}
\usepackage[colorlinks=false,urlbordercolor=white]{hyperref}
\newcommand{\say}[1]{\text{\hspace{1.5cm}#1}}
  
\tikzset{sgplattice/.style={inner sep=1pt,norm/.style={red!50!blue},char/.style={blue!50!black},
  lin/.style={black!50}},cnj/.style={black!50,yshift=-2.5pt,left=-1pt of #1,scale=0.5,fill=white}}
\newtheorem{theorem}{Theorem}[section]
\newtheorem*{theorem*}{Theorem}
\newtheorem{lemma}[theorem]{Lemma}
\newtheorem*{lemma*}{Lemma}
\newtheorem{corollary}[theorem]{Corollary}
\newtheorem{conjecture}[theorem]{Conjecture}

\makeatletter
\def\thmhead@plain#1#2#3{%
  \thmname{#1}\thmnumber{\@ifnotempty{#1}{ }\@upn{#2}}%
  \thmnote{ {\the\thm@notefont#3}}}
\let\thmhead\thmhead@plain
\makeatother
\newcommand{\brac}[1]{\left( #1 \right)}
\newcommand\bfrac[2]{\left(\frac{#1}{#2}\right)}
\usepackage{color}

\author{Tolson Bell\thanks{thbell@cmu.edu. Research supported in part by NSF Graduate Research Fellowship grant DGE 2140739.}~~and Alan Frieze\thanks{frieze@cmu.edu. Research supported in part by NSF grant DMS 1952285}\\Department of Mathematical Sciences\\Carnegie Mellon University\\Pittsburgh, PA 15213\\U.S.A.}
\title{$O(1)$ Expected Insertion for Random Walk $d$-ary\\Cuckoo Hashing up to the Load Threshold\thanks{A preliminary version of this paper appeared in the proceedings of the Foundations of Computer Science (FOCS) 2024 Conference.}}
\date{December 3, 2025}
\begin{document}
\maketitle
\thispagestyle{empty}
\begin{abstract}\large
The random walk $d$-ary cuckoo hashing algorithm was defined by Fotakis, Pagh, Sanders, and Spirakis to generalize and improve upon the standard cuckoo hashing algorithm of Pagh and Rodler. Random walk $d$-ary cuckoo hashing has low space overhead, guaranteed fast access, and fast in practice insertion time. In this paper, we give a theoretical insertion time bound for this algorithm. More precisely, for every $d\ge 3$ random hashes, let $c_d^*$ be the sharp threshold for the load factor at which a valid assignment of $cm$ objects to a hash table of size $m$ exists with high probability. We show that for any $d\ge 3$ hashes and load factor $c<c_d^*$, the expectation of the random walk insertion time is $O(1)$, that is, a constant depending only on $d$ and $c$ but not $m$.
\end{abstract}
\newpage\clearpage\pagenumbering{arabic}
\tableofcontents
\section{Introduction}
\subsection{Problem Statement and Theorem}\label{intro1}
In random walk $d$-ary cuckoo hashing, the goal is to store a set of objects $X$ where $X\subseteq W$ for a universe of elements $W$ in a hash table with table slots $Y$ given $d$ hash functions $h_1,\dotsc,h_d:W\rightarrow Y$. Following previous literature, we will take each hash function to be chosen independently and uniformly at random from all functions from $X$ to $Y$. When a new object $x_1$ is inserted, a uniformly random $i_1\in[d]$ is chosen, and $x_1$ is placed into position $h_{i_1}(x_1)$. If $h_{i_1}(x_1)$ was already occupied, we remove its previous occupant, $x_2$, and reinsert $x_2$ by the same algorithm (choosing a new $i_2\in[d]$ and putting $x_2$ into $h_{i_2}(x_2)$). This iterative algorithm terminates when we insert an object into an empty slot.

An object $x$ is queried by checking $h_1(x),\dotsc,h_d(x)$, which takes constant time for constant $d$. If we want to remove $x$, we simply delete it from its slot in the hash table. Thus access and deletion are both guaranteed to be fast.

Let $n=|X|$ and $m=|Y|$. We represent the hash functions as a bipartite graph with vertex set $(X,Y)$, and for each $x\in X$, edges from $x$ to $h_1(x),\dotsc,h_d(x)$. For a set $W\subseteq X$, we let $N(W)$ denote its set of neighbors in $Y$. An analogous definition is assumed for $Z\subseteq Y$. Finally, we replace $N(\{u\})$ by $N(u)$ for singleton sets.

For the insertion process to terminate, it must be true that there is an assignment of every object to a slot such that no slot has more than one object assigned to it and every object $x$ is assigned to $h_i(x)$ for some $1\le i\le d$. This can be represented as a matching of size $n$ in the bipartite graph. We know by Hall's Theorem that such a matching exists if and only if $|N(W)|\ge |W|$ for every $W\subseteq X$.

Unless explicitly noted otherwise, all asymptotics in this paper are written for $n\rightarrow\infty$ (or equivalently, $m\rightarrow\infty$) with $n=cm$ for fixed $d\in\mb{N}$ and fixed constant load factor $c\in(0,1)$. For instance, one could say that access and deletion need to query $\Theta(1)$ slots in the worst case, as our asymptotics suppress all factors depending on $c$ and $d$. We use ``with high probability'' to mean with probability $1-o(1)$ as $m,n\rightarrow\infty$ for fixed $d$ and $c$. When we refer to ``with high probability'' events, we are always talking about things that happen with high probability over the choice of the random hash function on the already inserted objects. That is, we do not use this phrase to refer to the likelihood of events dependent on the initial hash of the newly inserted object or on the progression of the random walk.

There is a sharp threshold $c_d^*$, called the \underline{load threshold}, for a matching of size $n$ to exist in the bipartite graph; that is, there is a constant $c_d^*$ such that if $c<c_d^*$ then there exists a matching with high probability and if $c>c_d^*$ then there with high probability does not exist a matching \cite{xorsatthresholds,FP12,FM12}.

Our result, which will be stated more precisely in Section \ref{prelim}, is the following:
\begin{theorem}\label{theorem}
Assume that we have $d\ge 3$, $c<c_d^*$, and $n=cm$. Then with high probability over the random hash functions, we have that the expected insertion time for the random walk insertion process is $O(1)$.\\Additionally, under the same conditions, for any constant $C\ge 0$, there is a constant $C'=C'(C,c,d)=\Theta(1)$ such that for sufficiently large $n$ and all $\ell\in\mb{N}$, the probability of the random walk insertion process taking more than $\ell$ steps is at most $C'\ell^{-C}$.
\end{theorem}
Throughout the paper, we will use $\Theta(1)$ to denote constants that may depend on $d$ or $c$ but do not depend on $n$ or $m$. For instance, Theorem \ref{theorem} means that the expected insertion time (conditioned on a with high probability event) is a constant depending only on $d$ and $c$ but not $n$ or $m$. We do not attempt to optimize this constant.

By insertion time, we mean the number of reassignments, that is, the number of times we move an object to a different one of its hash functions during the insertion. Our freedom to ignore, say, the cost of reading the $O(\log n)$-length value of $h_1(x)$ is sometimes known as the ``Word RAM'' model of computation. This assumption is consistent with all previous research on cuckoo hashing and hash tables more generally.

We do not explicitly consider deletions in this paper (consider building the hash table only), but our results are robust to any sequence of oblivious (not adaptive to the hash values) deletions and insertions of new elements (excluding re-insertions of deleted elements) as long as the total number of elements ever inserted stays below $cm$ for some $c<c_d^*$. 

Note that we are required to take our statement to only hold with high probability over the choices of hash functions, as there is a non-zero chance that the hash functions will not have any valid assignment of objects to slots (will fail Hall's condition) and thus will have infinite insertion time. Our ``with high probability'' statements are true with high probability over not just one element's insertion but over the entire process of building a cuckoo hash table of $cn$ elements for $c<c_d^*$. Therefore, we do get that with high probability the expected time to build a cuckoo hash table of $cn$ elements for $c<c_d^*$ is $O(n)$.  

\subsection{Applications and Relation to Previous Literature}
Standard cuckoo hashing was invented by Pagh and Rodler in 2001 \cite{PR01} and has been widely used in both theory and practice. Their formulation, though originally phrased with two hash tables, is essentially equivalent to the case $d=2$ of the algorithm described here. They showed that for all $c<c_2^*=0.5$, one can get $O(1)$ expected insertion time, an analysis that was extended by Devroye and Morin \cite{PR01,DM03}. Thus, cuckoo hashing is a data structure with $O(1)$ average-case insertion, $O(1)$ worst-case access and deletion, and only twice the amount of space that the elements themselves take up.

Cuckoo hashing can be seen as the ``average-case'' or ``random graph'' version of the ``online bipartite matching with replacements'' problem, with BFS insertion corresponding to the ``shortest augmenting walk'' algorithm. Take any bipartite graph with $V=(X,Y)$ that contains a matching of size $|X|$. If elements of $X$ and their incident edges arrive online, the amortized BFS insertion time was proven to be $O(\log^2(n))$ \cite{onlinebipartitelog2}. The lower bound is $\Omega(\log(n))$ \cite{onlinebipartiteintro}, which is matched if the vertex arrival order is randomized \cite{randomorderbipartite}. The previous paragraph shows that if the graph itself is random rather than worst-case, this $\Omega(\log(n))$ amortized insertion time bound is with high probability reduced to $\Theta(1)$.

$d$-ary cuckoo hashing was invented by Fotakis, Pagh, Sanders, and Spirakis in 2003 \cite{FPSS03}. The main advantage of increasing $d$ above 2 is that the load threshold increases. Even going from $d=2$ to $d=3$, the threshold $c_d^*$ goes from 0.5 to $\approx 0.918$, that is, with just one more hash function, we can utilize 91\% of the hash table instead of 49\%. The corresponding tradeoff is that the access time increases linearly with $d$. $d$-ary cuckoo hashing, also called generalized cuckoo hashing or improved cuckoo hashing, ``has been widely used in real-world applications'' \cite{widelyused}. 

The exact value for $c_d^*$ for all $d\ge 3$ was discovered via independent works by a number of authors \cite{xorsatthresholds,FP12,FM12}. This combinatorial problem of finding the matching threshold in these random bipartite graphs (which can also be viewed as random $d$-uniform hypergraphs) is directly related to other problems like $d$-XORSAT \cite{xorsatthresholds} and load balancing \cite{GWloadthresholds,FKPloadthresholds}.  

The primary insertion algorithm analyzed by Fotakis, Pagh, Sanders, and Spirakis was not random walk insertion, but rather was BFS insertion. In BFS insertion, instead of selecting a random $i_1\in[d]$ and hashing $x_1$ to $h_{i_1}(x_1)$, the algorithm finds the insertion walk minimizing the number of reassignments. In other words, $i_1,\dotsc,i_\ell\in[d]$ are chosen such that $\ell$ is minimized, where $x_1$ is to be hashed to $h_{i_1}(x_1)$, the removed object $x_2$ is to be hashed to $h_{i_2}(x_2)$, and so on until $h_{i_\ell}(x_\ell)$ is an empty slot. While BFS insertion requires more overhead to compute in practice, it is easier to analyze theoretically than random walk insertion. Fotakis, Pagh, Sanders, and Spirakis proved that BFS insertion only requires $O(1)$ expected reassignments for load factor $c$ when $d\ge 5+3\log(c/(1-c))$ \cite{FPSS03}. Our Corollary \ref{BFS} (which on its own has a shorter proof than Theorem \ref{theorem}) shows that this result extends to all $d\ge 3$ and $c<c_d^*$.  

Fotakis, Pagh, Sanders, and Spirakis also introduced the insertion algorithm we study, random walk insertion, describing it as ``a variant that looks promising in practice'', since they did not theoretically bound its insertion time but saw its strong performance in experiments \cite{FPSS03}. Random walk insertion requires no extra space overhead or precomputation. In a 2009 survey on cuckoo hashing, Mitzenmacher's first open question was to prove theoretical bounds for random walk insertion, calling random walk insertion ``much more amenable to practical implementation'' and ``usually much faster'' than BFS insertion \cite{Mitzen}. Insertion algorithms other than random walk or BFS have also been proposed, which have proven $O(n)$ total insertion time for $O(n)$ elements with high probability \cite{newinsertion} or more evenly distributed memory usage \cite{wearmin}. However, random walk insertion ``is currently the state-of-art method'' \cite{newinsertion}.

For load factors somewhat below the load threshold and $d\ge 8$, the random walk expected insertion time was proven to with high probability be polylogarithmic by Frieze, Melsted and Mitzenmacher in 2009 \cite{PolylogAlan}. Fountoulakis, Panagiotou, and Steger then extended this result to show polylogarithmic expected insertion time holds with high probability for all $d\ge 3$ and $c<c_d^*$. The exponent of their logarithm is anything greater than $1+b_d$, where $b_d=\frac{d+\log(d-1)}{(d-1)\log(d-1)}$ \cite{FPS13}. Our proof uses techniques and lemmas from these two papers.

The average-case insertion time for hash tables is expected to be $O(1)$, however, not poly-logarithmic. The first $O(1)$ random walk insertion bound was proven by Frieze and Johansson, who showed that for any load factor $c$, there exists some $d$ such that there is $O(1)$ expected insertion time with high probability for $d$ hashes at load factor $c$ \cite{FriezeJohansson}. However, their bounds only hold for large $d$ and load factors significantly less than the load threshold, $c=1-O_{d\rightarrow\infty}(\log(d)/d)$, while it had been shown that $c_d^*=1-(1+o_{d\rightarrow\infty}(1))(e^{-d})$ \cite{xorsatthresholds, FP12, FM12}.

To obtain a result that works for lower $d$, Walzer used entirely different techniques to prove $O(1)$ expected random walk insertion time with high probability up to the ``peeling threshold'', a load factor that is a lower number than the load threshold for any $d\ge 3$. The strongest result here is in the case $d=3$, where Walzer gets $O(1)$ expected insertion up to load factor $c=0.818$, compared to the optimal value $c_3^*=0.918$. Walzer pointed out that there was no $d\ge 3$ for which $O(1)$ insertion was known up to the load threshold, saying, ``Given the widespread use of cuckoo hashing to implement compact dictionaries and Bloom filter alternatives, closing this gap is an important open problem for theoreticians'' \cite{Walzer}.

Theorem \ref{theorem} is the first result to get $O(1)$ expected random walk insertion with high probability up to the load threshold for any $d\ge 3$, and works for all $d\ge 3$. The state of the art results are summarized in the tables below:\\
\begin{center}
\begin{tabular}{|c?c?c?c|}
\hline
$d$ & $c_d^*$ & $O(1)$ expected insertion & Insertion time at\\
& & up to load factor$\dotsc$ & $c=(1-\e)c_d^*~\forall~\epsilon>0$\\
\hline
2 &\footnotemark[1]0.5 &\footnotemark[1]0.5 &\footnotemark[1]$O(1)$\\
3 &\footnotemark[2]0.918 &\footnotemark[3]0.818 &\footnotemark[5]$O(\log^{3.664}(n))$\\
4 &\footnotemark[2]0.977 &\footnotemark[3]0.772 &\footnotemark[5]$O(\log^{2.547}(n))$\\
5 &\footnotemark[2]0.992 &\footnotemark[3]0.702 &\footnotemark[5]$O(\log^{2.152}(n))$\\
6 &\footnotemark[2]0.997 &\footnotemark[3]0.637 &\footnotemark[5]$O(\log^{1.946}(n))$\\
Large &\footnotemark[2]$1-(1+o_{d\rightarrow\infty}(1))(e^{-d})$ &\footnotemark[4]$1-O_{d\rightarrow\infty}(\frac{\log d}{d})$ &\footnotemark[5]$O(\log^{1+(\log d)^{-1}+O_{d\rightarrow\infty}(1/d)}(n))$\\\hline
\end{tabular}     

Prior work: \footnotemark[1]{\cite{PR01,DM03}} \footnotemark[2]{\cite{xorsatthresholds,FP12,FM12}}
\footnotemark[3]{\cite{Walzer}}
\footnotemark[4]{\cite{FriezeJohansson}}
\footnotemark[5]{\cite{FPS13}}
\-\vspace{0.5cm}
\begin{tabular}{|c?c?c?c|}
\hline
$d$ & $c_d^*$ & $O(1)$ expected insertion & Insertion time at \\
& & up to load factor$\dotsc$ & $c=(1-\e)c_d^*~\forall~\epsilon>0$\\
\hline
2 &\footnotemark[1]0.5 &\footnotemark[1]0.5 &\footnotemark[1]$O(1)$\\
3 &\footnotemark[2]0.918 &\footnotemark[6]0.918&\footnotemark[6]$O(1)$ \\
4 &\footnotemark[2]0.977 &\footnotemark[6]0.977 &\footnotemark[6]$O(1)$ \\
5 &\footnotemark[2]0.992 &\footnotemark[6]0.992 &\footnotemark[6]$O(1)$\\
6 &\footnotemark[2]0.997 &\footnotemark[6]0.997 &\footnotemark[6]$O(1)$\\
Large &\footnotemark[2]$1-(1+o_{d\rightarrow\infty}(1))(e^{-d})$ &\footnotemark[6]$1-(1+o_{d\rightarrow\infty}(1))(e^{-d})$ &\footnotemark[6]$O(1)$\\\hline
\end{tabular}\-     

Bounds after our work: \footnotemark[6]Theorem \ref{theorem}
\end{center}     

\section{Preliminaries}\label{prelim}

Our techniques to prove Theorem \ref{theorem} build off the techniques of Fountoulakis, Panagiotou, and Steger \cite{FPS13}, who showed expansion-like properties of the bipartite hashing graph that hold with high probability. The main new ingredient is the introduction of recursively defined ``bad'' sets $B_i$ for $i\in\mb{N}$. In this section, we will give some intuition for the overall proof structure and will more precisely state our results.


\subsection{The Bipartite Graph and Matchings}     

We will study the form of the random walk where at each object removal, we choose a random one of the $d-1$ other hashes for the object that was just evicted (not returning it to the spot it was just evicted from). In Section \ref{BFSSection}, we will show that proving the expected run time of this non-backtracking random walk is $O(1)$ also proves the same of the random walk that chooses any one of the $d$ hashes each time (including the one it was just removed from). Section \ref{BFSSection} will also show an $O(1)$ expected insertion time for the BFS insertion for all $d\ge 3$ and all $c<c_d^*$.     

Let $\mc{M}$ be the starting matching of size $n-1$ just before we insert the $n$th element. We can think of $\mc{M}$ as turning the bipartite graph into a directed graph. In the bipartite graph, each object $x\in X$ is in an edge with a slot $y\in Y$ if $y$ is $h_i(x)$ for some $1\le i\le d$. In the directed graph, the edge between $x$ and $y$ is oriented $y\rightarrow x$ if $x$ is matched to slot $y$, while it is oriented $x\rightarrow y$ if $x$ is not matched to slot $y$. The cuckoo hashing procedure can be thought of as a random walk on this directed graph (with the random walk also changing the directions of some edges as it progresses).     

Let $U\subseteq Y$ be the set of open spots in the hash table, which stays the same at each time step while the algorithm is running (as the algorithm terminates when it hits an open slot).     

Our proof only relies on expansion-like properties of the bipartite graph on $(X,Y)$ that hold with high probability. In particular, given the random bipartite graph, our result holds for any arbitrary starting matching $\mc{M}$ of objects to slots. As our expectation is over the hash values of the object being inserted, one fact we do need is that the hash values of this object being inserted are random among all slots in $Y$, after $\mc{M}$ is determined.

This is also the reason that we can deal with deletions, so long as the total number of elements ever inserted is at most $cm$. As long as the bipartite graph of the hashes of those $cm$ elements satisfies our desired properties, it is alright if some of those elements are not present in our table.

In other words, our theorem could be stated in more detail as follows:

\begin{theorem*}[\textbf{\ref{theorem}}]
Assume that we have $d\ge 3$, $c<c_d^*$, and $n=cm$. We will condition on an event $\mc{A}$ that relates solely to the hashes of the $n$ objects and occurs with probability $1-o(1)$ over uniformly random hash functions. Let $i\le n$ and let $x\in X$ be the $i$-th element being inserted. If $\mc{A}$ occurs, then for any matching $\mc{M}$ of the first $i-1$ elements to slots of the hash table that is independent of the hash values of $x$, we have that the expectation (over the hash values of $x$ and the choices of the random walk) of the insertion time for the random walk insertion process on $x$ is $O(1)$.

Additionally, under the same conditions, for any constant $C\ge 0$, there is a constant $C'=C'(C,c,d)=\Theta(1)$ such that for sufficiently large $n$ and all $\ell\in\mb{N}$, the probability of the random walk insertion process taking more than $\ell$ steps is at most $C'\ell^{-C}$.
\end{theorem*}
Corollary \ref{tailbounds} tells us that furthermore, if $\mc{A}$ occurs, then for any constant $C_6\ge 0$, there is a constant $C_7=C_7(C_6,c,d)$ such that for sufficiently large $n$ and all $\ell\in\mb{N}$, the probability (over the hash values of $x$) of the random walk insertion process taking more than $\ell$ steps is at most $C_7\ell^{-C_6}$.  

For convenience, we will consider inserting the $n$-th element throughout the paper, which we imagine inserting into a random slot in $Y$ before determining the rest of its hash values. $\mc{A}$ can be taken to be one event over all $n$ insertions; or, in other words, every ``with high probability'' statement in our proof is about bipartite graph structures that persist when elements are removed from $X$ (in Lemmas \ref{allbutdelt}, \ref{Upperboundnbrs}, \ref{fewcycles}, \ref{FPSneighbors}, and \ref{reachingsmall}). Thus, our result implies an $O(n)$ expected time to build the hash table of $n$ elements online.     

Starting from some $x\in X$, we will use the convention that a \underline{walk of length $i$} means that we do $i$ reassignments, which corresponds to a walk of length $2i$ in the bipartite graph $(X,Y)$. Let $W_{+i}'(x)\subseteq X$ be the set of all possible endpoints of a walk of length $i$ starting from $x$ under a particular matching $\mc{M}$. Therefore, $|W_{+i}'(x)|\le (d-1)^i$, as we have $d-1$ choices of assignment at each step (at the first step, $x$ is banned from choosing the slot that it is matched to under $\mc{M}$).     

The $B_i$ will be defined based on counting the number of ``good'' elements in $W_{+i}'(x)$. Landing on an unoccupied slot $u\in U$ within $i$ steps is certainly good (as our process will be finished), so we do want to make sure that we are counting those walks which finish within $i$ steps towards our count of the good elements in $W_{+i}'(x)$. Intuitively, we want to imagine that a walk which hits an unoccupied slot $u\in U$ after $j$ steps for some $1\le j\le i$ continues for $i-j$ more steps after it hits $u$, so $u$ should count $(d-1)^{i-j}$ times as a good element of $W_{+i}'(x)$. For instance, if $x$ has one neighbor $u\in U$, we want $u$ to contribute $(d-1)^{i-1}$ dummy elements to $W_{+i}'(x)$. If there were also a different walk from $x$ that hit that same $u$ on the $j$th reassignment for some $1\le j\le i$, then the same $u$ would also contribute $(d-1)^{i-j}$ additional dummy elements, and so on.      

Formally, we accomplish this as follows: for every $i\in\mb{N}$ and $x\in X$, let the set $\mc{U}_i(x)$ be a set of dummy elements, that is, newly-introduced elements that are not in $X$ and are pairwise disjoint for every $x\in X$ and $i\in\mb{N}$. We will take \[
|\mc{U}_i(x)|=\sum_{j=1}^i(\#\text{walks from $x$ that hit $U$ on the $j$th reassignment})(d-1)^{i-j}.
\]      

Then we define $W_{+i}(x)=W_{+i}'(x)\sqcup\mc{U}_i(x)$. For $S\subseteq X$, we can similarly define $W_{+i}(S)=\bigcup_{x\in S}W_{+i}(x)$. We also define $W_{+\le i}(x)=\bigcup_{j=0}^i W_{+j}(x)$ and $W_{+\le i}(S)$ analogously for $S\subseteq X$. 

Similarly, for $x\in X$ and $j\in\mb{N}$, let $W_{-\le j}(x)$ be defined to equal $\{w\in X:x\in\bigcup_{k=0}^jW_{+k}(w)\}$, that is, the set of elements that could reach $x$ in at most $j$ steps. Unlike for $W_{+\le j}(x)$, there is no need for dummy elements in this definition. $W_{-\le j}(S)$ is defined as expected for $S\subseteq X$.

The \underline{BFS distance}, or \underline{distance}, of an object $w\in X$ from an object $x\in X$ under $\mc{M}$ is the minimal $i$ such that $w\in W_{+i}(x)$. We can define BFS distances involving sets in the natural way, by minimizing over elements of those sets. We can similarly define the BFS distance of a slot $y\in Y$ from an object $x\in X$ as 1 plus the BFS distance from $x$ to $N(y)$. For example, $\{h_1(x),\dotsc,h_d(x)\}$ is exactly the set of slots at BFS distance 1 from $x$. Slots with no hash functions to them (isolated vertices in the bipartite graph) or otherwise unreachable slots can be assumed to have infinite distance.     

\begin{lemma}[(Corollary 2.3 of \cite{FPS13})]\label{allbutdelt}
Let $d\ge 3$ and assume $n=cm$ for $c<c_d^*$. Then with high probability, we have that for any matching $\mc{M}$ and any $\alpha=\Theta(1)>0$, there exists $M=\Theta(1)$ such that for the unoccupied vertices $U$ of $Y$, we have that at most $\alpha n$ of the vertices of $X$ have BFS distance $>M$ to $U$.
\end{lemma}
Lemma \ref{allbutdelt} was first proven by the inventors of $d$-ary cuckoo hashing, but only under the weaker condition $d\ge 5+3\log(c/(1-c))$ for $n=cm$ \cite{FPSS03}. (Note logarithms are natural unless denoted otherwise.) Corollary 2.3 of \cite{FPS13} extended this lemma to all $d\ge 3$ and $c<c_d^*$. Some intuition for Lemma \ref{allbutdelt} will be given in Section \ref{changes}.     

Let $\alpha>0$ be sufficiently small (but still $\Theta(1)$, to be set later) and take the corresponding $M=\Theta(1)$ as in Lemma \ref{allbutdelt}. For our starting matching $\mc{M}$, let $G$ be all vertices of $X$ of BFS distance at most $M$ from $U$. When we start at a vertex $g\in G$, we have at least a $(d-1)^{-M}$ chance that our random walk will finish in at most $M$ more steps. (That is, there is at least a $(d-1)^{-M}$ chance that our random walk will be the BFS walk, which has length $\le M$.) Intuitively, this gives that the expected length on a random walk that stays inside $G$ at every time $t$ is at most $M(d-1)^M+M=\Theta(1)$ (though some technicalities arise due to the changing matching as the walk progresses). This shows intuitively that it suffices to only focus on the ``worst'' $\alpha n$ vertices for some small $\alpha=\Theta(1)>0$.  

\subsection{``Backtracking'' Random Walk Insertion}

Throughout this paper, we will study a form of random walk insertion where an evicted object chooses uniformly at random one of its other $d-1$ hash values to insert at next. We will frequently use in our proof that we have $(d-1)^i$ possible walks of length $i$ from a given starting point of the insertion algorithm. However, some implementations of random walk insertion may simply choose to insert an object at any one of its $d$ hash values, including the one it was just evicted from.
\begin{corollary}
Theorem \ref{theorem} still holds for the form of random walk insertion where each object chooses uniformly among its $d$ hash functions for re-insertion at each step.
\end{corollary}
\begin{proof}
In order to prove this, we will give a coupling from the random walk with backtracking, to the random walk without backtracking but with some delays. Essentially, every time the walk backtracks, we can imagine that it just stayed in the same spot for the same amount of time that the backtracking took. We will show that the expected time ``wasted'' by this backtracking simply multiplies the expected random walk time by at most a $O(1)$ factor. Therefore, if the non-backtracking walk had $O(1)$ expected time, then the backtracking walk also has $O(1)$ expected time.    

Every time we are at an object $x_{i+1}$ on step $i+1$ of the random walk and choose the hash that $x_{i+1}$ was just evicted from, that means that we return again to the previous object $x_i$. Essentially, we will charge this time backtracked as time wasted while at $x_i$. So, at every step $x_i$, we unveil how many steps will be ``wasted'' only to eventually end up back at $x_i$ through backtracking, and charge those to $x_i$ right then before going to the new object $x_{i+1}$.

To do this, we want to upper bound the probability of backtracking to return to $x_i$. This does not include the probability that we cycle around on new hashes to return to $x_i$, so we are only thinking of returning through the exact same hashes we leave from $x_i$ on.    

To return to $x_i$ by backtracking in exactly $2t$ steps (at time $i+2t$), we need to choose $t$ of those steps to be backtracks. Each of those $t$ steps has an independent $\frac 1d$ probability of indeed being a backtrack, while the other $t$ cannot be a backtrack, which has an independent $\frac{d-1}{d}$ probability for each. Therefore, the probability that we return to $x_i$ by backtracking in $2t$ steps is upper bounded by $\binom{2t}{t}d^{-t}\left(\frac{d-1}{d}\right)^t\le\left(\frac{4(d-1)}{d^2}\right)^t$. Note that this does already include the situation where we backtrack to $x_i$ in fewer than $2t$ steps, and then backtrack again to reach $x_i$ again exactly $2t$ steps after time $i$.

Then for all $d\ge 3$, the expected delay at $x_i$ from backtracking is at most $\sum_{t=1}^\infty 2t\left(\frac{4(d-1)}{d^2}\right)^t=O(1)$ as desired, as the sum is convergent.

The exponential tail bounds on this sum also mean that the tail bounds of Theorem \ref{theorem} also transfer. In particular, if the non-backtracking random walk had probability $p$ of taking more than $\ell$ steps, then for the backtracking random walk to take more than $5\ell$ steps, it must be true that there were at least $2\ell$ backtracking steps. As each step independently has probability at most $1/d\le 1/3$ to be backtracking, the chance of at least $2\ell$ of the first $5\ell$ steps to be backtracking (before any conditioning) is $e^{-O(\ell)}$. Thus, the probability of the backtracking random walk taking more than $5\ell$ steps is $p+e^{-O(\ell)}$. Our desired bounds on $p$ are subexponential in $\ell$, so this extra $e^{-O(\ell)}$ factor, as well as the change from $\ell$ to $5\ell$, are lower-order terms.
\end{proof}     

\subsection{Paper Outline}

In Section \ref{nbrs}, we will show an upper bound on the number of hashes that any set of slots receives. Iterating this will prove in Lemma \ref{nbrsofnbrs} that, with high probability, for any $j\in\mb{N}$ and any $S\subseteq X$ with $|S|\le n/12$, we have \[|W_{-\le j}(S)|\le\left(3d\log\left(\frac{n}{|S|}\right)\right)^j|S|.
\]

As the random walk progresses, the matching of objects to slots changes as we perform evictions. Section \ref{changes} explains why this does not present a problem for our analysis. 

Section \ref{expansion} is our longest and most technical section. It begins by giving a lower bound on the number of distinct slots hashed to by a set of objects in Lemma \ref{FPSneighbors}. Applying Lemma \ref{FPSneighbors} gives lower bounds on $|W_{+j}(S)|$ that hold for any $S\subseteq X$. This lemma proves to be a keystone of our proof, as the lower bounds on $|W_{+j}(S)|$ can be iteratively built up to give bounds on the likelihood of ending up in one set when starting from another. Intuitively, if $|W_{+j}(S)|$ is near its upper bound of $(d-1)^j|S|$, then the random walks starting in $S$ do not concentrate on any small set of vertices, which helps our analysis.

In Section \ref{chaining}, we define the sets $B_i$, where we iteratively define \[
B_i=\{x\in X:\text{at least }2(d-1)^ii^{-1.5}\text{ walks of length $i$ from $x$ end in }B_{i-1}\}.
\]In other words, if we are outside of $B_i$, we have at least $1-2i^{-1.5}$ probability that in $i$ steps we will be outside of $B_{i-1}$. So, if the random walk begins outside of $B_i$, it is likely to iteratively progress from $X\setminus B_i$ to $X\setminus B_{i-1}$, and so on, to eventually reach the $G$ of Lemma \ref{allbutdelt}. The fact that $\sum_{i=1}^\infty 2i^{-1.5}$ converges means that we can achieve an arbitrarily small constant probability of this progression failing on any step.

Lemma \ref{Bi4} from Section \ref{expansion} quickly implies that $|B_i|\le(d-1)^{-i^2/4}n$. Section \ref{chaining} continues on to show that walks starting ``sufficiently far'' from $B_i$ have probability at least 0.97 of finishing in $O(i^2)$ steps. 

Section \ref{improvedreaching} directly improves Lemma \ref{nbrsofnbrs}, with a more technical proof giving a stronger bound on $|W_{-\le j}(S)|$. This was not needed for anything before Section \ref{improvedreaching}, but is needed in Section \ref{finalproofsection}, which finishes the proof of $O(1)$ insertion by improving the ``with probability at least 0.97'' statement to an expected insertion time. Roughly, we can show that if a walk starting in $X\setminus B_i$ fails to finish in $O(i^2)$ steps, then we are still likely to be outside of $X\setminus B_{9i}$ and can attempt another run. 

Section \ref{tailsection} proves stronger tail bounds on the insertion time of the random walk, that is, an upper bound on the probability that the random walk will take at least $\ell$ steps. 

Section \ref{BFSSection} extends our work to show an $O(1)$ expected insertion time for a BFS insertion, a very different insertion proceedure to random walk.

Finally, Section \ref{futureworksection} discusses possible future improvements on our results.

\section{Bounding the Number of Walks to any Set}\label{nbrs}
To show that reaching some bad set is unlikely, we want to upper bound the probability of reaching some small set, which we can later combine with a proof of bad sets being small. To accomplish this, we need to bound the number of neighbors that a small set can have.
\begin{lemma}\label{Upperboundnbrs}
For any $d\ge 3$ and $c<c_d^*$, we have with high probability that there is not a set $Z\subseteq Y$ with $|Z|\le n/12$ such that $|N(Z)|\ge 3d\log\left(\frac{n}{|Z|}\right)|Z|$.
\end{lemma}
\begin{proof}
First, imagine fixing $Z\subseteq Y$, then randomly choosing the edges of our graph. Let $e(Z)$ be the number of edges incident to $Z$. Our bipartite graph has $dn$ edges, and each has an independent $|Z|/m\le |Z|/n$ chance of landing in $|Z|$. Thus, $e(Z)$ is stochastically dominated by the binomial random variable $\mathrm{Bin}(dn,|Z|/n)$, and so $\\\mb{E}(e(Z))\le d|Z|$. By standard Chernoff bounds,
\[
\mb{P}\left(e(Z)\ge 3d\log\left(\frac{n}{|Z|}\right)|Z|\right)\le\bfrac{e}{3\log(n/|Z|)}^{3d|Z|\log(n/|Z|)}\le e^{-3d|Z|\log(n/|Z|)}=\left(\frac{|Z|}{n}\right)^{3d|Z|}.
\]
Then 
\begin{align*}
\mb{P}&\left(\exists~Z\subseteq Y\text{ s.t. }|N(Z)|\ge 3d\log\left(\frac{n}{|Z|}\right)|Z|\right)\\
&\le\sum_{i=1}^{n/12}\binom mi\left(\frac{i}{n}\right)^{3di}
\le\sum_{i=1}^{n/12}\left(\frac{2en}{i}\right)^i\left(\frac{i}{n}\right)^{3di}=\sum_{i=1}^{n/12}\brac{2e\bfrac{i}{n}^{3d-1}}^i\\
&\le\sum_{i=1}^{\log^2(n)}2e\bfrac{\log^2(n)}{n}^2+\sum_{i=\log^2(n)}^{n/12}\brac{2e\bfrac{1}{12}^{2}}^{\log^2(n)}=o(1/n).
\end{align*}
\end{proof}
Now, for $x\in X$ and $j\in\mb{N}$, let $W_{-\le j}(x)$ be defined to equal $\{w\in X:x\in\bigcup_{k=0}^jW_{+k}(w)\}$, that is, the set of elements that could reach $x$ in at most $j$ steps.
\begin{lemma}\label{nbrsofnbrs}
For any $d\ge 3$ and $c<c_d^*$, we have with high probability that for any matching $\mc{M}$ for any $j\in\mb{N}$ and any $S\subseteq X$ with $|S|\le n/12$, we have $|W_{-\le j}(S)|\le\left(3d\log\left(\frac{n}{|S|}\right)\right)^j|S|$.
\end{lemma}
\begin{proof}
We will assume that the conclusion of Lemma \ref{Upperboundnbrs} holds, as it does with high probability. We can then prove this lemma inductively as a corollary of Lemma \ref{Upperboundnbrs}.      

We see that Lemma \ref{nbrsofnbrs} is true for $j=0$. Then note that $W_{-\le j}(S)=W_{-\le 1}(W_{-\le (j-1)}(S))=N(Z)$ where $Z\subseteq Y$ is the spots occupied by $W_{-\le (j-1)}(S)$, which thus has the same cardinality of $W_{-\le (j-1)}(S)$. 

So using Lemma \ref{Upperboundnbrs}, we have
\begin{multline*}
|W_{-\le j}(S)|\le3d\log\left(\frac{n}{|W_{-\le (j-1)}(S)|}\right)|W_{-\le (j-1)}(S)|\le 3d\log\left(\frac{n}{|S|}\right)|W_{-\le (j-1)}(S)|\\
\le\left(3d\log\left(\frac{n}{|S|}\right)\right)^j|S|,
\end{multline*}
as desired.  

Note that if we ever have $|W_{-\le (j-1)}(S)|\ge n/12$ (so Lemma \ref{Upperboundnbrs} can't be applied), then we have $|W_{-\le j}(S)|\le 3d\log\left(\frac{n}{|S|}\right)|W_{-\le (j-1)}(S)|$ anyway, as the right side of the equation is then more than $n$.
\end{proof} 

Lemma \ref{nbrsofnbrs} will be strong enough for our work in the next few sections, including the technical Section \ref{expansion}. After Section \ref{expansion}, in Section \ref{improvedreaching} we will prove Lemma \ref{reachingsmall}, which is a more technical improvement on Lemma \ref{nbrsofnbrs} needed for our final proof. 

\section{The Changing Matching}\label{changes}
As we noted above, Fountoulakis, Panagiotou, and Steger proved the following lemma, which we will use as a black box:
\begin{lemma*}[\textbf{\ref{allbutdelt}} (Corollary 2.3 of \cite{FPS13})]
Let $d\ge 3$ and assume $n=cm$ for $c<c_d^*$. Then with high probability, we have that for any matching $\mc{M}$ and any $\alpha=\Theta(1)>0$, there exists $M=\Theta(1)$ such that for the unoccupied vertices $U$ of $Y$, we have that at most $\alpha n$ of the vertices of $X$ have BFS distance $>M$ to $U$.
\end{lemma*}
The above lemma comes as a corollary of their following theorem:
\begin{lemma}[(Theorem 2.2 of \cite{FPS13})]\label{contraction}
Let $d\ge 3$ and assume $n=cm$ for $c<c_d^*$. Then with high probability, there exists a $\delta>0$ such that for every nonempty $R\subseteq Y$, we have $|\{x\in X:N(x)\subseteq R\}|<(1-\delta)|R|$.
\end{lemma}
Note for comparison that $|\{x\in X:N(x)\subseteq R\}|\le|R|$ for every $R\subseteq Y$ is exactly the requirement for a matching to exist. This is essentially saying that for $c<c_d^*$, we beat Hall's bound by a constant factor for all sets as $n\rightarrow\infty$. 

In other words, you could consider the parameter of the graph $\xi=\max_{R\subseteq Y}\frac{|\{x\in X:N(x)\subseteq R\}|}{|R|}$. The definition of $c_d^*$ tells us that for any $c>c_d^*$, we with high probability have $\xi>1$,  while for any $c<c_d^*$, we with high probability have $\xi\le 1$. The theorem above says that for any $c<c_d^*$, there exists an $\delta=\delta(c)$ such that we with high probability have $\xi\le 1-\delta$.  

When we first start our random walk by inserting the $n$-th element, it is inserted into a random slot in $Y$, independent of any previous hashes or actions taken by the cuckoo hashing process when inserting the previous elements. This fact is critical to our proof. Interestingly, the paper of \cite{FPS13} does not use this fact; their result would hold true even if the initial hash of the element were adversarially chosen: 

\begin{theorem}[(Theorem 1.2; Lemma 2.7 of \cite{FPS13})]\label{FPStheorem}
Assume that we have $d\ge 3$, $c<c_d^*$, and $n=cm$. Then with high probability over the random hash functions, we have that the expected insertion time for the random walk insertion process is $O(\log^{1+b_d}(n))$, where $b_d=\frac{d+\log(d-1)}{(d-1)\log(d-1)}$.
\end{theorem}
It is useful for us to use as a black-box that we can have $O(\mathrm{poly}\log n)$ insertion time even when starting from an arbitrary starting hash in the graph. In other words, at any point in the insertion process, conditioned on any prior events in the insertion process (and still assuming the ``with high probability'' facts about the underlying graph), the expected time from that time until the random walk finishes is $O(\log^{1+b_d}(n))$. $b_d\le 3$ for all $d\ge 3$, so this is $O(\log^4(n))$.  

This proof that $O(\log n)$ expected insertion time from a given step holds conditioned on any prior events also can be shown to follow from our work here: Lemma \ref{expsmallalli} will show that with high probability it is true that for any matching (even one that might have been changed over the course of the walk) that our bad set $B_i$ has that $B_i=\emptyset$ when $i=C'\sqrt{\log{n}}$ with a sufficiently high constant $C'$. It is more notationally convenient to explain away the changes to the matching now, so we can hereafter treat the matching as fixed.  

As the random walk progresses, the matching changes from $\mc{M}$, as some elements are moved to different spots. Again picturing $\mc{M}$ as assigning directions to the edges of the bipartite graph, as we only change the direction of edges that we move along, we see that the only time that the change in the matching may affect our random walk might need to worry about the changing matching is if the walk cycles; that is, if an object $x\in X$ is reached twice by the random walk.  

For this purpose, we will define a special set $\mc{C}\subseteq X$, which we can think of as the vertices near cycles significantly shorter than $\log(n)$. Formally, let $z=(10\log(n))^{0.9}$ and let $S_{Cyc}\subseteq X$ be the set of vertices that are on a cycle of length $z$ or less. Then we define $\mc{C}=W_{-\le z}(S_{Cyc})$. 

\begin{lemma}\label{fewcycles}
For any $d\ge 3$ and $c<c_d^*$, we have with high probability over the choice of random hashes that $|\mc{C}|<n^{0.3}$.
\end{lemma}
Before proving this lemma, we will explain how it allows us to deal with changes to $\mc{M}$. We break into two cases based on whether our random walk starts inside $\mc{C}$ or not.

If we start inside $\mc{C}$, we will simply use the $O(\mathrm{poly}\log n)$ bound of Theorem \ref{FPStheorem}, which as explained above works conditioned on any starting hash. Since the probability of starting in $\mc{C}$ is at most $n^{-0.7}$, this adds an $O(1)$ term to our expected run time. 

In the second case, we note that our random walk will not repeat any elements within the first $z/2$ steps, so we can take the matching as fixed for the first $z/2$ steps. We will show in Corollary \ref{tailbounds} that, for a fixed matching, the probability of taking more than $z/2$ steps is $O(z^{-5})\le O((\log n)^{-4})$ as well. So, if the random walk starting outside of $\mc{C}$ reaches $z/2$ steps in length, then we can again revert to the $O(\log^4 n)$ bound (which again holds conditioned on any prior events in the random walk) while only adding an $O(1)$ term to the expected run time.  

Therefore, this subsection shows that we do not need to worry about any changes to the matching from $\mc{M}$, as the cases that remain to be proven only include cases that do not involve any cycling in the random walk. So, for the rest of this paper, we can consider the cuckoo hashing insertion procedure to be a random walk on the fixed directed graph given by $\mc{M}$.

\begin{proof}[Proof of Lemma \ref{fewcycles}]
Fix $\ell\in\mb{N}$ and consider the cycles of length $2\ell$ in the bipartite graph. Each has the form $(x_1,y_1,x_2,y_2,\dotsc,x_\ell,y_\ell)$ for some $x_1,\dotsc,x_\ell\in X$ and $y_1,\dotsc,y_\ell\in Y$, where $x_i$ hashes to both $y_i$ and $y_{i-1}$ (with $x_1$ also hashing to $y_\ell$). There are at most $n^\ell m^\ell$ ordered sets of vertices $(x_1,y_1,x_2,y_2,\dotsc,x_\ell,y_\ell)$. The probability that all required hashes will be chosen is at most $\left(\frac{d(d-1)}{m^2}\right)^\ell\le d^{2\ell}m^{-2\ell}$. Thus, the expected number of cycles of length $2\ell$ in the bipartite graph is at most $n^\ell m^\ell d^{2\ell}m^{-2\ell}<d^{2\ell}$.

Then the expected number of cycles of length at most $z$ is at most $\\ \sum_{\ell=1}^{z/2} d^{2\ell}\le d^{z+1}= o(d^{\log(n)/(100d)})= o(n^{0.1})$. Markov's inequality gives that with high probability there are less than $n^{0.1}$ cycles of length at most $z$.

Each of these cycles has at most $z$ vertices on it, so $|S_{Cyc}|<n^{0.1}z<n^{0.2}$ for sufficiently large $n$. 

Then we apply Lemma \ref{nbrsofnbrs} to say that \[|\mc{C}|\le \left(3d\log\left(\frac{n}{|S_{Cyc}|}\right)\right)^z|S_{Cyc}|<\left(3d\log(n)\right)^zn^{0.2}<n^{0.3}.\]
\end{proof}

\section{Expansion from any vertex set}\label{expansion}
\subsection{Lower bounds on $|W_{+j}(S)|$}
The following lemma was proven by Fountoulakis, Panagiotou, and Steger: 

\begin{lemma}[(Proposition 2.4 of \cite{FPS13})]\label{FPSneighbors}
Let $d\ge 3$ and $c<c_d^*$. For any $1\le s<|X|/d$, define \[
p_s=\begin{cases}0&\text{if }s\le\log\log(n)\\\frac{\log_d((d-1)e^d)}{\log_d(|X|/(ds))}&\text{if }\log\log(n)<s<|X|/d\end{cases}
\] With high probability, we have that for all $S\subseteq X$ with $|S|<|X|/d$ that \[|N(S)|\ge (d-1-p_{|S|})|S|.\]
\end{lemma}
Some facts to note here are that $p_s\ge 0$ and $p_s$ is monotonically increasing with $s$. Also, for $s\le |X|/(e^{1000d})$, which we will generally be able to assume by Lemma \ref{allbutdelt}, we have $p_s<0.01$.   

Fountoulakis, Panagiotou, and Steger applied Lemma \ref{FPSneighbors} iteratively, and the term in the numerator of Lemma \ref{FPSneighbors} ended up becoming the exponent of the logarithm in their $O(\mathrm{poly}\log n)$ run-time bound.  

Because it is so critical for our paper, we provide a proof of Lemma \ref{FPSneighbors} here, reproducing the proof of Fountoulakis, Panagiotou, and Steger \cite{FPS13}.

\begin{proof}
For a given $S\subseteq X$ and $T\subseteq Y$, we have that \[\mb{P}(N(S)\subseteq T)=\left(\frac{|T|}{m}\right)^{d|S|}.
\]
Fix an $s$ such that $1\le s\le n$. For there to be an $S\subseteq X$ with $|S|=s$ that fails the lemma, it must have $N(S)\subseteq T$ for a $T$ with $|T|=(d-1-p_s)s$. So, for $s>\log\log(n)$, we have
\begin{align*}
\mb{P}(\exists~S\subseteq X\text{ with }|S|=s\text{ failing the lemma})&\le\binom ns\binom{m}{(d-1-p_s)s}\left(\frac{(d-1-p_s)s}{m}\right)^{ds}
\\&\le\left(\frac{ne}{s}\right)^s\left(\frac{em}{(d-1-p_s)s}\right)^{(d-1-p_s)s}\left(\frac{(d-1-p_s)s}{m}\right)^{ds}
\\&\le\left(\frac{cs^{p_s}e^{(d-p_s)}(d-1-p_s)^{(1+p_s)}}{m^{p_s}}\right)^s
\\&\le\left(\left(\frac{s(d-1)}{em}\right)^{p_s}(ce^d(d-1))\right)^s
\\&\le\left(\left(\frac{ds}{|X|}\right)^{p_s}(ce^d(d-1))\right)^s
\\&\le c^s,
\end{align*}where the last step comes from the definition of $p_s$. For $s\le\log(m)/(d^2)$, we can note that for $S$ to fail the lemma we must in fact have $|T|<(d-1)s$, so as $(d-1)s$ is an integer we need $|T|\le(d-1)s-1=(d-1-1/s)s$. So, applying the above process with $1/s$ in the place of $p_s$, we have for $s\le\log(m)/(d^2)$ that \begin{align*}
\mb{P}(\exists~S\subseteq X\text{ with }|S|=s\text{ failing the lemma})&\le\binom ns\binom{m}{(d-1-1/s)s}\left(\frac{(d-1-1/s)s}{m}\right)^{ds}
\\&\le\left(\left(\frac{s(d-1)}{em}\right)^{1/s}(ce^d(d-1))\right)^s
\\&\le\left(\frac{\log(m)}{edm}\right)\left(ce^d(d-1)\right)^{\log(m)/(d^2)}
\\&\le\left(\frac{\log(m)}{edm}\right)2^{\log(m)}
\\&\le o(m^{-0.2})
\end{align*}
Then summing over all $s$, the probability that there exists some $S\subseteq X$ that fails our lemma is at most \[
\sum_{s=1}^{\log(m)/(d^2)}o(m^{-0.2})+\sum_{s=\log(m)/(d^2)}^{cm/d}c^s\le o(m^{-0.19})+O(m^{log(c)/(d^2)}).
\]So, as the load factor $c$ satisfies $c<c_d^*<1$, we see that with high probability there is no such $S$, as desired.
\end{proof}  

For our purposes, we want to give a lower bound on $|W_{+j}(S)|$ using this lemma. Note that we have a natural upper bound of $|W_{+j}(S)|\le(d-1)^j|S|$. Additionally, recalling that $W_{+j}'(S)$ is the version of $W_{+j}(S)$ but with dummy elements excluded, we have $|W_{+j}'(S)|\le d|W'_{+j}(S)|\le d|W_{+j}(S)|$.
\begin{lemma}\label{evenmorebasicWjlower}
Let $d\ge 3$, $c<c_d^*$, and $p_s$ be defined as in Lemma \ref{FPSneighbors}. With high probability, for every matching $\mc{M}$ and any $S\subseteq X$ and $j\in\mb{N}$ with $|S|<(d-1)^{-j-2}n$, we have that  \[|W_{+\le j}(S)|\ge (d-1-p_{|W_{+\le (j-1)}(S)|})|W_{+\le (j-1)}(S)|.\]
\end{lemma}
\begin{proof}
Recall that $|W_{+(j-1)(S)}|=|W'_{+(j-1)}(S)|+|\mc{U}_{j-1}(S)|$.

First, note that it is always true that $|\mc{U}_j(x)|\ge (d-1)|\mc{U}_{j-1}(x)|$, as each element of $\mc{U}_{j-1}(x)$ contributes $d-1$ elements to $\mc{U}_j(x)$. It may be true that $\mc{U}_j(x)$ also has additional elements if an element $x\in W'_{j-1}(S)\subseteq X$ has an unoccupied slot in $N(x)$.

We then claim that \begin{equation}\label{singlelower}|W_{+\le j}(S)|\ge|N(W'_{+\le (j-1)}(S))|+(d-1)|\mc{U}_{\le (j-1)}(S)|.\end{equation} The first term counts $|W'_{+\le j}(S)|$ as well as every walk from $S$ that first hits an unoccupied slot on exactly the $j$th step, while the second term counts every walk from $S$ that hits an unoccupied slot before the $j$th step.

We can apply Lemma \ref{FPSneighbors} on top of \eqref{singlelower} to say that 
\begin{align*}|W_{+\le j}(S)|&\ge|N(W'_{+\le (j-1)}(S))|+(d-1)|\mc{U}_{\le (j-1)}(S)|\\&\ge (d-1-p_{|W'_{+\le (j-1)}(S)|})|W'_{+\le (j-1)}(S)|+(d-1)|\mc{U}_{\le (j-1)}(S)|\\&\ge(d-1-p_{|W'_{+\le (j-1)}(S)|})|W_{+\le (j-1)}(S)|\\&\ge(d-1-p_{|W_{+\le (j-1)}(S)|})|W_{+\le (j-1)}(S)|\end{align*}as desired.
\end{proof}

\begin{lemma}\label{basicWjLower}
Let $d\ge 3$, $c<c_d^*$, and $p_s$ be defined as in Lemma \ref{FPSneighbors}. With high probability, for every matching $\mc{M}$ and any $S\subseteq X$ and $j\in\mb{N}$ with $|S|<0.01(d-1)^{-j-2}n$, we have that $|W_{+j}(S)|\ge(d-1-p_{(d-1)^j|S|})^{j-1}(d-2-p_{(d-1)^j|S|})|S|$.
\end{lemma}
\begin{proof}
Note first that the condition $|S|<0.01(d-1)^{-j-2}n$ means that $p_{(d-1)^j|S|}<d-2$, which is necessary for our factors to be non-negative.

Applying Lemma \ref{evenmorebasicWjlower} inductively gives
\[|W_{+\le j}(S)|\ge |S|\prod_{k=0}^{j-1}(d-1-p_{(d-1)^{k+1}|S|})\ge(d-1-p_{(d-1)^j|S|})^j|S|,\]
using that $p_s$ is monotonically increasing with $s$. Then,
\begin{align}
 |W_{+j}(S)|&\ge|W_{+\le j}(S)|-|W_{+\le (j-1)}(S)|\nonumber\\
 &\ge(d-1-p_{(d-1)^j|S|})|W_{+\le (j-1)}(S)|-|W_{+\le (j-1)}(S)|\nonumber\\
 &\ge(d-2-p_{(d-1)^j|S|})|W_{+\le (j-1)}(S)|\label{add}\\
 &\ge(d-2-p_{(d-1)^j|S|})(d-1-p_{(d-1)^{j-1}|S|})^{j-1}|S|\nonumber\\
 &\ge(d-2-p_{(d-1)^j|S|})(d-1-p_{(d-1)^j|S|})^{j-1}|S|\nonumber
\end{align}
as desired.
\end{proof}
As $p_s$ is small for sufficiently small $s$, Lemma \ref{basicWjLower} intuitively tells us that $|W_{+j}(S)|$ is roughly within a constant factor, $\frac{d-2}{d-1}$, of the upper bound of $|W_{+j}(S)|\le(d-1)^j|S|$. This constant-factor difference is still not strong enough for some of our applications. Lemma \ref{advancedWjLower} then shows that relatively little of this loss appears at higher $j$.
\begin{lemma}\label{advancedWjLower}
Let $d\ge 3$, $c<c_d^*$, and $p_s$ be defined as in Lemma \ref{FPSneighbors}. For any $S\subseteq X$ and any $j\in\mb{N}$ with $|S|<(d-1)^{-j}n/e^{1000d}$, we have with high probability that \[|W_{+j}(S)|\ge (d-1)|W_{+(j-1)}(S)|-|S|-2.1p_{(d-1)^j|S|}|W_{+(j-1)}(S)|.\]
\end{lemma}
\begin{proof}
We again start from Lemma \ref{evenmorebasicWjlower}. In other words, this means that if you take any ordering of the $d|W_{+\le (j-1)}(S)|$ hashes leaving $W_{+\le (j-1)}(S)$, there are at    most $(1+p_{|W_{+\le (j-1)}(S)|})|W_{+\le (j-1)}(S)|$ hashes that hit a table slot already hit by another hash in $W_{+\le (j-1)}(S)$, which we could call repeat hashes.     

Considering ordering the hashes where the ones from $S$ come first, then the ones from $W_{+1}(S)$, and so on until $W_{+(j-1)}(S)$. We see that for every $1\le k<j-1$, by the definition of $W_{+k}(S)$, every object in $W_{+k}(S)$ sends at least one of its $d$ hashes into a slot occupied by an element of $W_{+(k-1)}(S)$, giving a repeat at every element in $W_{+k}(S)$. That shows that from $W_{+1}(S)$ to $W_{+(j-2)}(S)$ we have \[\ge\bigcup_{k=1}^{j-2}|W_{+k}(S)|\ge|W_{+\le (j-1)}(S)|-|W_{+(j-1)}(S)|-|S|\] repeats.     

Therefore, the number of the $d|W_{+(j-1)}(S)|$ hash functions from $W_{+(j-1)}(S)$ that can go to a slot already hashed to is at most the number of repeats remaining, which is at most \begin{multline*}
    (1+p_{|W_{+\le (j-1)}(S)|})|W_{+\le (j-1)}(S)|-(|W_{+\le (j-1)}(S)|-|W_{+(j-1)}(S)|-|S|)\\=|W_{+(j-1)}(S)|+|S|+p_{|W_{+\le (j-1)}(S)|}|W_{+\le (j-1)}(S)|.
\end{multline*} Each new slot that is hashed to gives a corresponding element of $W_{+j}(S)$, so 
\begin{align*}
|W_{+j}(S)|&\ge d|W_{+(j-1)}(S)|-(|W_{+(j-1)}(S)|+|S|+p_{|W_{+\le (j-1)}(S)|}|W_{+\le (j-1)}(S)|)\\&\ge(d-1)|W_{+(j-1)}(S)|-|S|-p_{|W_{+\le (j-1)}(S)|}|W_{+\le (j-1)}(S)|\\
&\ge(d-1)|W_{+(j-1)}(S)|-|S|-p_{(d-1)^j|S|}|W_{+\le (j-1)}(S)|\\
&\ge(d-1)|W_{+(j-1)}(S)|-|S|-p_{(d-1)^j|S|}(|W_{+(j-1)}(S)|+|W_{+\le j-2}(S)|)\\
&\ge(d-1)|W_{+(j-1)}(S)|-|S|-p_{(d-1)^j|S|}\brac{|W_{+(j-1)}(S)|+\frac{|W_{+(j-1)}(S)|}{d-2-p_{(d-1)^{j-1}|S|}}},
\\&\say{using \eqref{add}}\\
&\ge(d-1)|W_{+(j-1)}(S)|-|S|-p_{(d-1)^j|S|}\brac{|W_{+(j-1)}(S)|+\frac{|W_{+(j-1)}(S)|}{d-2.01}},
\\&\say{ using $|W_{+(j-1)}(S)|\le(d-1)^{j-1}|S|<|X|/(e^{1000d})$}\\
&\ge(d-1)|W_{+(j-1)}(S)|-|S|-2.1p_{(d-1)^j|S|}|W_{+(j-1)}(S)|
\end{align*}
as desired.
\end{proof}     

\subsection{Avoiding small sets through expansion}
Now that we have proven lower bounds on $|W_{+j}(S)|$ for any set, we will use these bounds to go in a different direction and show that for a set $S$, relatively few elements will have many walks to $S$. In other words, we create sets $B_j(S)$, which consist of elements $x\in X$ from which we have a high likelihood of being in $S$ after $j$ steps.

Our definition of ``high likelihood'' will change in each of four lemmas, Lemmas \ref{Bi1}, \ref{Bi2}, \ref{Bi3}, and \ref{Bi4}. These four lemmas will define four different sets: $B_j^{(1)}(S)$, $B_j^{(2)}(S)$, $B_j^{(3)}(S)$, and $B_j^{(4)}(S)$ respectively. These lemmas will iteratively ``bootstrap'' off of each other to give stronger results each time. Additionally, we will see that the way Lemma \ref{Bi4} builds off Lemma \ref{Bi3} is analogous to the way Lemma \ref{Bi2} builds off Lemma \ref{Bi1}.

Unlike bounding $|W_{-\le j}(S)|$, which includes all objects that have at least one walk of length $j$ to $S$, we will only include objects which have at least some fraction of their walks of length $j$ reaching $S$. Correspondingly, while $|W_{-\le j}(S)|$ must grow with $j$, when we set the parameters here, we will see that $|B_j(S)|$ will actually shrink with $j$; if $j$ is large, there are very few objects that, as a start point, have a high likelihood of being in $S$ after a $j$ steps. You can imagine this as saying something about the mixing of our random walk, as it shows the random walk is unlikely to concentrate on a small set after a while.     

A sketch of the basic argument goes like this: if you take a sufficiently large set $Q$, then Lemmas \ref{basicWjLower} and \ref{advancedWjLower} show that $|W_{+j}(Q)|$ is large. In particular, this means that there are many distinct endpoints for a walk of length $j$ starting in $Q$. Only $|S|$ of the distinct endpoints are in $S$, so if $|W_{+j}(Q)|\gg|S|$, then it is unlikely for a walk of length $j$ starting in $Q$ to end up in $S$. Therefore, this $Q$ cannot be $B_j(S)$, so $|B_j(S)|<|Q|$.

For this to work, it is also necessary to know how many of the walks could concentrate on the same endpoints.

Formally now: for any set $S\subseteq X$ and a given matching $\mc{M}$, define $B_i^{(1)}(S)\subseteq X$ as follows. An object $x$ is in $B_i^{(1)}(S)$ if and only if at least 3/4 of the $(d-1)^{\log^2(i)}$ walks of length $\log^2(i)$ from $x$ end in an object $z\in W_{+\log^2(i)}(x)$ such that at least $\frac{150}{i}$ proportion of the $(d-1)^{i-\log^2(i)}$ walks of length $i-\log^2(i)$ starting at $z$ end in $S$.     

In other words, we define \[
S'_1(i,S)=\{z\in X:\ge 150(d-1)^{i-\log^2(i)}/i\text{ walks of length }i-\log^2(i)\text{ from }z\text{ end in }S\}
\]and then \[
B_i^{(1)}(S)=\{x\in X:\ge 3(d-1)^{\log^2(i)}/4\text{ walks of length }\log^2(i)\text{ from }x\text{ end in }S'_1(i,S)\}.
\]
\begin{lemma}\label{Bi1}
For any $d\ge 3$ and $c<c_d^*$, there exists a $C_1=\Theta(1)$ such that we have the following with high probability. For any matching $\mc{M}$, any $S\subseteq X$, and any $i\in\mb{N}$ with $i\ge C_1$ and $|S|\le(d-1)^{-i^2/5}n$, it is true that $|B_i^{(1)}(S)|<(d-1)^{-0.9i}|S|$.
\end{lemma}
Before we prove Lemma \ref{Bi1}, we will note that the definition of $B_i^{(1)}(S)$ is relatively complicated because of the introduction of this intermediate set $S'_1(i,S)$. We will overcome this two-phase iteration in Lemma \ref{Bi1point5}, resulting in a simpler definition of $B_i^{(2)}(S)$ in Lemma \ref{Bi2}. However, this definition will still be not as strong as we need for later sections. Lemma \ref{Bi3} will define $B_i^{(3)}(S)$ similarly to $B_i^{(1)}(S)$, but with a stronger bound on the $150(d-1)^{i-\log^2(i)}/i$ term. Lemma \ref{Bi1point5} will then be applied again to simplify the definition, giving $B_i^{(4)}(S)$.
\begin{proof}
Let $Q\subseteq X$ such that $|Q|=(d-1)^{-0.9i}|S|$. We will prove that there must exist some $x\in Q$ such that $x\notin B_i^{(1)}(S)$, therefore proving that no such $Q$ can equal $B_i^{(1)}(S)$ and thus $|B_i^{(1)}(S)|<(d-1)^{-0.9i}|S|$.     

In fact, we will show this by showing that, starting from a uniformly random point $x\in Q$, there is at least probability $\ge\frac 14$ that after $\log^2(i)$ steps, we are at a point $z$ such that more than $1-\frac{150}{i}$ proportion of the $(d-1)^{i-\log^2(i)}$ walks of length $i-\log^2(i)$ do not end in $S$.     

We note that for every $\log^2(i)\le j\le i$, we have that \begin{align*}
|W_{+j}(Q)|&\ge (d-1)|W_{+(j-1)}(Q)|-|Q|-2.1p_{(d-1)^j|Q|}|W_{+(j-1)}(Q)|\say{By Lemma \ref{advancedWjLower}}
\\&\ge (d-1)|W_{+(j-1)}(Q)|-\frac{3|W_{+\log^2(i)}(Q)|}{(d-1)^{\log^2(i)}}-2.1p_{(d-1)^{i-i^2/5}n}|W_{+(j-1)}(Q)|\\&\ge (d-1)|W_{+(j-1)}(Q)|-\frac{|W_{+(j-1)}(Q)|}{i^2}-2.1\frac{\log_d((d-1)e^d)}{(i^2/5-i)\log_d(d-1)-1}|W_{+(j-1)}(Q)|\\&\ge\left(d-1-\frac{2.1(15(d-1))}{i^2}\right)|W_{+(j-1)}(Q)|,
\end{align*}
again using that $15(d-1)>\frac{5\log_d((d-1)e^d)}{\log_d(d-1)}$ for all $d\ge 3$.   

Applying this iteratively, we get that \begin{multline*}
    |W_{+i}(Q)|\ge\left(d-1-\frac{32(d-1)}{i^2}\right)^{i-\log^2(i)}|W_{+\log^2(i)}(Q)|\\
    \ge(d-1)^{i-\log^2(i)}\left(1-\frac{32}{i^2}\right)^i|W_{+\log^2(i)}(Q)|
    \ge(d-1)^{i-\log^2(i)}\left(1-\frac{32}{i}\right)|W_{+\log^2(i)}(Q)|
\end{multline*}
Every one of the $(d-1)^{i-\log^2(i)}|W_{+\log^2(i)}(Q)|$ walks of length $i-\log^2(i)$ leaving $W_{+\log^2(i)}(Q)$ ends in an element of $W_{+i}(Q)$, so the above tells us that at most $32/i$ proportion of the walks of length $i-\log^2(i)$ leaving $W_{+\log^2(i)}(Q)$ end at an object that was also the endpoint of a previous such walk; that is, there are at most $(d-1)^{i-\log^2(i)}\left(\frac{32}{i}\right)|W_{+\log^2(i)}(Q)|$ repeats.

Now, 
\[|S|= (d-1)^{0.9i}|Q|\le (d-1)^{-0.05 i}(d-1)^{i-\log^2(i )}|Q|.\]
This implies that of the $(d-1)^{i-\log^2(i)}|W_{+\log^2(i)}(Q)|$ walks of length $i-\log^2(i)$ from $W_{+\log^2(i)}(Q)$, at most a $\left(\frac{32}{i}+(d-1)^{-0.05i}\right)\le\frac{33}{i}$ proportion end in $S$, as if we order these walks we can have at most $|S|$ ones hit an element of $S$ not hit by a previous path, plus the repeats. Then, the Markov inequality tells us that less than $\frac 14$ of the elements in $W_{+\log^2(i)}(Q)$ have at least a $\frac{150}{i}$ proportion of their walks ending in $S$, and thus (recalling the definition of $S'_1(i,S)$ before the start of Lemma \ref{Bi1}), \begin{equation}|S'_1(i,S)\cap W_{+\log^2(i)}(Q)|<|W_{+\log^2(i)}(Q)|/4.\label{withinBi1}\end{equation}

Now, we note that
\begin{align*}
|W_{+\log^2(i)}(Q)|&\ge (d-1-p_{(d-1)^{\log^2(i)}|Q|})^{\log^2(i)-1}(d-2-p_{(d-1)^{\log^2(i)}|Q|})|Q|\say{by Lemma \ref{basicWjLower}}\\
&\ge\left(d-1-\frac{\log_d((d-1)e^d)}{\log_d((d-1)^{i^2/5}/d)}\right)^{\log^2(i)-1}\left(d-2-\frac{\log_d((d-1)e^d)}{\log_d((d-1)^{i^2/5}/d)}\right)|Q|,\\
&\say{as $\log^2(i)|Q|<|S|\le(d-1)^{-i^2/5}n$}\\
&\ge\left(d-1-\frac{15(d-1)}{i^2}\right)^{\log^2(i)-1}\left(d-2-0.1\right)|Q|,\\&\say{using $15(d-1)\ge\frac{5\log_d((d-1)e^d)}{\log_d(d-1)}$ for all $d\ge 3$}\\
&\ge(d-1)^{\log^2(i)}\left(\frac{d-2.1}{d-1}\right)\left(1-\frac{15(\log^2(i)-1)}{i^2}\right)|Q|\\
&\ge  (d-1)^{\log^2(i)}|Q|/3.
\end{align*}

Combining the above with \eqref{withinBi1}, we get that \[|W_{+\log^2(i)}(Q)\setminus S'_1(i,S)|>(3/4)|W_{+\log^2(i)}(Q)|\ge (d-1)^{\log^2(i)}|Q|/4.\]
This finishes the proof, as it then must be true that more than $1/4$ of the $(d-1)^{\log^2(i)}|Q|$ walks of length $\log^2(i)$ leaving $Q$ must not end up in $S'_1(i,S)$, meaning that $Q$ cannot be $B_i^{(1)}(S)$.
\end{proof}

Now, we proceed to the second of four lemmas, where we can replace the two step ``constant probability of having $\frac{150}{i}$ probability of being in $S$'' with a pure $\frac{200}{i}$ probability of being in $S$.

For any set $S\subseteq X$, define $B_i^{(2)}(S)\subseteq X$ under a given matching $\mc{M}$ as follows. An object $x$ is in $B_i^{(2)}(S)$ if and only if at least a $\frac{200}{i}$ fraction of the $(d-1)^i$ walks of length $i$ starting at $x$ end in $S$.     

In other words, you could define \[B_i^{(2)}(S)=\{x\in X:\ge {200}(d-1)^i/i\text{ walks of length }i\text{ from }x\text{ end in }S\}.\]     

\begin{lemma}\label{Bi2}
For any $d\ge 3$ and $c<c_d^*$, there exists a $C_2=\Theta(1)$ such that we have the following with high probability. For any matching $\mc{M}$, any $S\subseteq X$, and any $i\in\mb{N}$ with $i\ge C_2$ and $(d-1)^{-i^{10}}n<|S|<(d-1)^{-i^2/4.5}n$, it is true that $|B_i^{(2)}(S)|<(d-1)^{-0.8i}|S|$. 
\end{lemma}
To prove this, we will first give a more general amplification lemma:
\begin{lemma}\label{Bi1point5}
Let $a,b,c\in\mb{N}$. Assume that we have that $A\subseteq X$ has the property that from every vertex outside $A$, an $a$-step block has probability at least 1/4 of reaching a vertex from which, after a further $c$ steps, the probability of lying outside $W_{-\le (ab)}(S)$ is at least $1-q$. If we start outside of $W_{-\le (ab)}(A)$, then the probability of being inside of $S$ in $a(b+1)+c$ steps is at most $0.75^b+q$.
\end{lemma}
\begin{proof}[Proof of Lemma \ref{Bi1point5}]
Let $x\notin W_{-\le(ab)}(A)$. 

Let $x_2$ be the position that we reach after $a$ steps from $x$. Since $x\notin A$, there is at least a 1/4 chance that $x_2$ has the property that there is at least a $1-q$ chance that we will be outside of $W_{-\le (ab)}(S)$ in a further $c$ steps, and thus outside $S$ in $a(b+1)+c$ steps.

If this 1/4 probability event does not occur, then as $x_2\notin A$ (assuming $b\ge 1$), we still have that after a further $a$ steps from $x_2$, there is again at least a 1/4 chance that we are at a point $x_3$ such that with probability $1-q$, we will be outside of $W_{-\le (ab)}(S)$ in a further $c$ steps from $x_3$, and thus again succeed in being outside of $S$ in $a(b+1)+c$ steps.

In this way, we see that we can iterate, and then for any $k\in\mb{N}$ such that $1\le k\le b$, it is true that after $ka$ steps, we have probability at least $1-\left(\frac 34\right)^k$ to reach a point $x'$ such that at least a $1-q$ fraction of walks from $x'$ are outside $W_{-\le(ab)}(S)$ in a further $c$ steps from $x'$.  

Conditioned on this $\ge 1-\left(\frac 34\right)^b$ probability event (of reaching such an $x'$ within $ab$ steps) succeeding, we have probability at least $1-q$ of not being in $S$ within $a(b+1)+c$ steps. Therefore, our total probability of being in $S$ after $a(b+1)+c$ steps is $0.75^b+q$.

\end{proof}
\begin{proof}[Proof of Lemma \ref{Bi2}]
We first combine Lemmas \ref{Bi1} and \ref{Bi1point5}. In Lemma \ref{Bi1point5}, take $a=\log^2(i-\log^4(i))$, $b=\log^4(i)/a$, $c=i-a(b+1)$, and $q=\frac{150}{i-\log^4(i)}$. We can then verify that $B_{i-\log^4(i)}^{(1)}(W_{-\le \log^4(i)}(S)))$ by definition satisfies the conditions of Lemma \ref{Bi1point5} with respect to $S$.

Then we have that \[0.75^b+q=0.75^{\log^4(i)/\log^2(i-\log^4(i))}+\frac{150}{i-\log^4(i)}=0.75^{\Omega(\log^2(i))}+\frac{150}{i-\log^4(i)}\le 200/i.\] Therefore, \[B_i^{(2)}(S)\subseteq W_{-\le \log^4(i)}(B_{i-\log^4(i)}^{(1)}(W_{-\le \log^4(i)}(S))).\]

Our task now is to show that show $|W_{-\le \log^4(i)}(B_{i-\log^4(i)}^{(1)}(W_{-\le \log^4(i)}(S)))|<(d-1)^{-0.8i}|S|$. First note that
\[
|W_{-\le \log^4(i)}(S)|\le\left(3d\log\left(\frac{n}{|S|}\right)\right)^{\log^4(i)}|S|\le\left(3di^{10}\right)^{\log^4(i)}(d-1)^{-i^2/4.5}n\le(d-1)^{-i^2/5}n
\]
for $i\ge C_2$, by Lemma \ref{nbrsofnbrs}. So
\begin{align*}
|B_{i-\log^4(i)}^{(1)}(W_{-\le \log^4(i)}(S)))|&\le(d-1)^{-0.9(i-\log^4(i))}|W_{-\le \log^4(i)}(S)|\say{by Lemma \ref{Bi1}}\\&\le(d-1)^{-0.85i}\left(3d\log\left(\frac{n}{|S|}\right)\right)^{\log^4(i)}|S|\say{by Lemma \ref{nbrsofnbrs}}
\\&\le (d-1)^{-0.85i}\left(3di^{10}\log(d-1)\right)^{\log^4(i)}|S|\say{as $(d-1)^{-i^{10}}n<|S|$}
\\&\le (d-1)^{-0.82i}|S|
\end{align*}
and thus 
\begin{align*}
&|W_{-\le \log^4(i)}(B_{i-\log^4(i)}^{(1)}(W_{-\le \log^4(i)}(S)))|\\
&\le\left(3d\log\left(\frac{n}{|B_{i-\log^4(i)}^{(1)}(W_{-\le \log^4(i)}(S))|}\right)\right)^{\log^4(i)}|B_{i-\log^4(i)}^{(1)}(W_{-\le \log^4(i)}(S))|,\quad\text{by Lemma \ref{nbrsofnbrs}}\\
\\&\le\left(3d\log\left(\frac{n}{(d-1)^{-0.82i}|S|}\right)\right)^{\log^4(i)}(d-1)^{-0.82i}|S|
\\&\say{as $(C\log(n/s))^rs$ is monotone increasing on $s\le ne^{-r}$,}\\&\say{and here $|B_{i-\log^4(i)}^{(1)}(W_{-\le \log^4(i)}(S))|\le (d-1)^{-0.82i}|S|\le ne^{-\log^4(i)}$.}
\\&\le\left(3d(i^{10}+0.82i)\log(d-1)\right)^{\log^4(i)}(d-1)^{-0.82i}|S|
\\&\le (d-1)^{-0.8i}|S|
\end{align*}
as desired.     
\end{proof}

We defined $B_i^{(2)}(S)$ to have $200/i$ probability of hitting $S$. Now, we want to strengthen this result by reducing this probability to $2/i^{1.5}$. We will now bootstrap Lemma \ref{Bi2} to a stronger failure probability, but first using the same definition as in Lemma \ref{Bi1} except with the $150/i$ probability replaced with $i^{-1.5}$.

For any set $S\subseteq X$, define $B_i^{(3)}(S)\subseteq X$ under a given matching $\mc{M}$ as follows. An object $x$ is in $B_i^{(3)}(S)$ if and only if at least 3/4 of the $(d-1)^{\log^2(i)}$ walks of length $\log^2(i)$ from $x$ end in an object $z\in W_{+\log^2(i)}(x)$ such that at least $i^{-1.5}$ proportion of the $(d-1)^{i-\log^2(i)}$ walks of length $i-\log^2(i)$ starting at $z$ end in $S$.     

In other words, you could define
\[S'_3(i,S)=\{z\in X:\ge (d-1)^{i-\log^2(i)}/i^{1.5}\text{ walks of length }i-\log^2(i)\text{ from }z\text{ end in }S\}\]and then \[B_i^{(3)}(S)=\{x\in X:\ge 3(d-1)^{\log^2(i)}/4\text{ walks of length }\log^2(i)\text{ from }x\text{ end in }S'_3(i,S)\}.\]
\begin{lemma}\label{Bi3}
For any $d\ge 3$ and $c<c_d^*$, there exists a $C_3=\Theta(1)$ such that we have the following with high probability. For any matching $\mc{M}$, any $S\subseteq X$, and any $i\in\mb{N}$ with $i\ge C_3$ and $(d-1)^{-2i^3}n<|S|<(d-1)^{-i^2/4.5}n$, it is true that $|B_i^{(3)}(S)|<(d-1)^{-0.7i}|S|$. 
\end{lemma}
\begin{proof}
This proof will follow much of the same structure of Lemma \ref{Bi1}, and will also use the result of Lemma \ref{Bi2}.

Let $Q\subseteq X$ such that $|Q|=(d-1)^{-0.7i}|S|$. As in Lemma \ref{Bi1}, we will prove that, starting from a uniformly random point $x\in Q$, there is at least probability $\ge\frac 14$ that after $\log^2(i)$ steps, we are at a point $z$ such that more than a $1-i^{-1.5}$ proportion of the $(d-1)^{i-\log^2(i)}$ walks of length $i-\log^2(i)$ do not end in $S$. This shows that there must be some $x\in Q$ such that $x\notin B_i^{(3)}(S)$, proving that $|B_i^{(3)}(S)|<(d-1)^{-0.7i}|S|$.

In the same way as in the proof of Lemma \ref{Bi1}, we see that \[
|W_{+\log^2(i)}(Q)|\ge(d-1)^{\log^2(i)}|Q|/3
\]and
\begin{equation}\label{another}
|W_{+j}(Q)|\ge\left(d-1-\frac{32(d-1)}{i^2}\right)|W_{+(j-1)}(Q)|
\end{equation}
for every $\log^2(i)\le j\le i$. Applying these iteratively, we get that \begin{align*}
|W_{+j}(Q)|&\ge (d-1)^{j-\log^2(i)}\left(1-\frac{32}{i^2}\right)^{j-\log^2(i)}|W_{+\log^2(i)}(Q)|\ge(d-1)^j\left(1-\frac{32}{i}\right)|Q|/3
\\&\ge(d-1)^j|Q|/4\ge(d-1)^{j-0.7i}|S|/4
\end{align*}
for every $\log^2(i)\le j\le i$. We also have that, for every $2i^{0.3}\le k\le i$, \begin{align*}
|W_{-\le 1}(B_k^{(2)}(S))|&\le 3d\log\left(\frac{n}{|B_k^{(2)}(S)|}\right)|B_k^{(2)}(S)|\say{by Lemma \ref{nbrsofnbrs}}\\
&\le 3d\log\left(\frac{n}{(d-1)^{-0.8k}|S|}\right)(d-1)^{-0.8k}|S|
\\&\say{by Lemma \ref{Bi2}, as $k\ge C_2$, and $|S|>(d-1)^{-2i^3}n\ge(d-1)^{-k^{10}}n$,}
\\&\say{\hspace{1cm}and $|S|<(d-1)^{-i^2/4.5}n\le(d-1)^{-k^2/4.5}n$}\\
&\le 3d\log\left(\frac{n}{(d-1)^{-0.8i}(d-1)^{-i^3}n}\right)(d-1)^{-0.8k}|S|\\
&\le 3d(\log(d-1))(i^3+0.8i)(d-1)^{-0.8k}|S|\\
\end{align*}
Comparing these two bounds, we get that for all $\log^2(i)\le j\le i-2i^{0.3}$, \begin{align*}
|W_{-\le 1}(B_{i-j}^{(2)}(S))|&\le3d(\log(d-1))(i^3+0.8i)(d-1)^{-0.8(i-j)}|S|\\
&\le (d-1)^{0.8j-0.8i+0.05i}|S|\say{for all $i\ge C_3$}\\
&\le (d-1)^{(j-1)-0.7i}|S|/(4i^3)\say{for all $i\ge C_3$}\\
&\le |W_{+(j-1)}(Q)|/(i^3)
\end{align*}
Now, consider the $(d-1)|W_{+(j-1)}(Q)|$ hashes leaving $W_{+(j-1)}(Q)$ (that is, excluding the hashes currently occupied by $W_{+(j-1)}(Q)$). If $j\le i-2i^{0.3}$, at most an $i^{-3}$ proportion of them end inside $B_{i-j}^{(2)}(S)$.

Additionally, by \eqref{another}, at most an $\frac{32}{i^2}$ proportion of them under any ordering land on a slot already hashed to by a previous hash. If we use $R$ to denote the set of hashes that land on a slot that another hash from $W_{+(j-1)}(Q)$ also lands on, we have that $R$ is at most a $\frac{64}{i^2}$ proportion of the total hashes. Counting separately the $i^{-3}$ proportion ending up inside $B_{i-j}^{(2)}(S)$, the at most $\frac{64}{i^2}$ proportion outside of $B_{i-j}^{(2)}(S)$ corresponding to $R$ thus has at least probability $1-\frac{200}{i-j}$ of landing outside of $S$ after $i-j$ more steps ($i$ total steps).

Of the remaining hashes that go to unique locations at each step, at most $|S|$ end up in $S$.

Therefore, of the $(d-1)^{i-\log^2(i)}|W_{+\log^2(i)}(Q)|$ walks of length $i-\log^2(i)$ from $W_{+\log^2(i)}(Q)$, every walk that lands in $S$ falls into one of the following categories:
\begin{itemize}
    \item At the first $j$ where its position is the same as the position of another one of the walks, it falls into $B_{i-j}^{(2)}(S)$
    \begin{itemize}
        \item At most an $i^{-3}$ proportion of walks for a given $\log^2(i)\le j\le i-2i^{0.3}$
        \item At  most an $64i^{-2}$ proportion of walks for a given $i-2i^{0.3}\le j\le i$
    \end{itemize}
    \item At the first $j$ where its position is the same as the position of another one of the walks, it does not fall into $B_{i-j}^{(2)}$
    \begin{itemize}
        \item At most an $\frac{64}{i^2}$ proportion of walks for a given $j$
        \item At most a $\frac{200}{i-j}$ proportion of the walks that fall into this category for this $j$ end up in $S$, if $i-j\ge i^{0.3}$
    \end{itemize}
    \item Has no $j$ such that its position at step $j$ is the same as the position of another one of the walks
    \begin{itemize}
        \item At most $|S|$ total walks
    \end{itemize}
\end{itemize}
Putting this together, the proportion that land in $S$ out of the of the $(d-1)^{i-\log^2(i)}|W_{+\log^2(i)}(Q)|$ walks of length $i-\log^2(i)$ from $W_{+\log^2(i)}(Q)$ is at most
\begin{align*}
{\frac{i-\log^2(i)}{i^{3}}}+&\frac{64}{i^2}\sum_{j=\log^2(i)}^{i-2i^{0.3}}{ \frac{200}{i-j}}+\frac{64}{i^2}\sum_{j=i-2i^{0.3}}^i 1+\frac{|S|}{(d-1)^{i-\log^2(i)}|W_{+\log^2(i)}(Q)|}\\
&\le i^{-2}+\frac{12800\log(i)}{i^2}+\frac{128i^{0.3}}{i^2}+\frac{|S|}{(d-1)^i|Q|/3}\\
&<i^{-1.5}/4.\say{for $i\ge C_3$}
\end{align*}
From here, we finish the proof as in Lemma \ref{Bi1}: The Markov inequality tells us that less than $\frac 14$ of the  elements in $W_{+\log^2(i)}(Q)$ have at least a $i^{-1.5}$ proportion of their walks ending in $S$, and thus (recalling the definition of $S'_3(i,S)$ before the start of Lemma \ref{Bi3}),
\[|S'_3(i,S)\cap W_{+\log^2(i)}(Q)|<|W_{+\log^2(i)}(Q)|/4\]
so then 
\[|W_{+\log^2(i)}(Q)\setminus S'_3(i,S)|>(3/4)|W_{+\log^2(i)}(Q)|\ge (d-1)^{\log^2(i)}|Q|/4.\]
This finishes the proof, as it then must be true that more than $1/4$ of the $(d-1)^{\log^2(i)}|Q|$ walks of length $\log^2(i)$ leaving $Q$ must not end up in $S'_3(i,S)$, meaning that $Q$ cannot be $B_i^{(3)}(S)$.
\end{proof}
Finally, we complete the analogy with $B_i^{(4)}(S)$, which will be defined for $B_i^{(3)}(S)$ in the way that $B_i^{(2)}(S)$ was for $B_i^{(1)}(S)$.

For any set $S\subseteq X$, define $B_i^{(4)}(S)\subseteq X$ under a given matching $\mc{M}$ as follows. An object $x$ is in $B_i^{(4)}(S)$ if and only if at least a $2i^{-1.5}$ proportion of the $(d-1)^i$ walks of length $i$ starting at $x$ end in $S$.     

In other words, you could define \[B_i^{(4)}(S)=\{x\in X:\ge 2(d-1)^ii^{-1.5}\text{ walks of length }i\text{ from }x\text{ end in }S\}.\]     

\begin{lemma}\label{Bi4}
For any $d\ge 3$ and $c<c_d^*$, there exists a $C_4=\Theta(1)$ such that we have the following with high probability. For any matching $\mc{M}$, any $S\subseteq X$, and any $i\in\mb{N}$ with $i\ge C_4$ and $(d-1)^{-i^3}n<|S|<(d-1)^{-i^2/4}n$, it is true that $|B_i^{(4)}(S)|<(d-1)^{-0.6i}|S|$.
\end{lemma}
\begin{proof}
We first combine Lemmas \ref{Bi3} and \ref{Bi1point5}, just as we did in Lemma \ref{Bi2}. Again, we apply Lemma \ref{Bi1point5} with $a=\log^2(i-\log^4(i))$, $b=\log^4(i)/a$ and $c=i-a(b+1)$. Now, $q=(i-\log^4(i))^{-1.5}$. We can then verify that $B_{i-\log^4(i)}^{(3)}(W_{-\le \log^4(i)}(S)))$ by definition satisfies the conditions of Lemma \ref{Bi1point5} with respect to $S$.

Then we have that \[0.75^b+q=0.75^{\log^4(i)/\log^2(i-\log^4(i))}+(i-\log^4(i))^{-1.5}=0.75^{\Omega(\log^2(i))}+(i-\log^4(i))^{-1.5}\le 2i^{-1.5}.\] Therefore, \[B_i^{(4)}(S)\subseteq W_{-\le \log^4(i)}(B_{i-\log^4(i)}^{(3)}(W_{-\le \log^4(i)}(S))).\]

Then \[
|W_{-\le \log^4(i)}(S)|\le\left(3d\log\left(\frac{n}{|S|}\right)\right)^{\log^4(i)}|S|\le\left(3di^3\right)^{\log^4(i)}(d-1)^{-i^2/4}n\le(d-1)^{-i^2/4.5}n
\]for $i\ge C_4$, by Lemma \ref{nbrsofnbrs}. So
\begin{align*}
|B_{i-\log^4(i)}^{(3)}(W_{-\le \log^4(i)}(S)))|&\le(d-1)^{-0.7(i-\log^4(i))}|W_{-\le \log^4(i)}(S)|\say{by Lemma \ref{Bi3}}\\&\le(d-1)^{-0.65i}\left(3d\log\left(\frac{n}{|S|}\right)\right)^{\log^4(i)}|S|\say{by Lemma \ref{nbrsofnbrs}}
\\&\le (d-1)^{-0.65i}\left(3di^3\log(d-1)\right)^{\log^4(i)}|S|\say{as $(d-1)^{-i^3}n<|S|$}
\\&\le (d-1)^{-0.62i}|S|
\end{align*}
and thus 
\begin{align*}
&|W_{-\le \log^4(i)}(B_{i-\log^4(i)}^{(3)}(W_{-\le \log^4(i)}(S)))|\\
&\le\left(3d\log\left(\frac{n}{|B_{i-\log^4(i)}^{(3)}(W_{-\le \log^4(i)}(S))|}\right)\right)^{\log^4(i)}|B_{i-\log^4(i)}^{(3)}(W_{-\le \log^4(i)}(S))|,\quad\text{by Lemma \ref{nbrsofnbrs}}\\
\\&\le\left(3d\log\left(\frac{n}{(d-1)^{-0.62i}|S|}\right)\right)^{\log^4(i)}(d-1)^{-0.62i}|S|
\\&\le\left(3d(i^{10}+0.62i)\log(d-1)\right)^{\log^4(i)}(d-1)^{-0.62i}|S|
\\&\le (d-1)^{-0.6i}|S|
\end{align*}
as desired.     
\end{proof}

Note that if we adjust the constants, we could repeat this process an arbitrary constant number of times, define a $B_i$ with $i^{-c}$ probability of hitting $S$ for any $c\in\mb{R}$, and still show this is smaller than $S$ by an exponential factor in $i$. But for our overall proof, the $2i^{-1.5}$ probability we have now obtained suffices.

\section{Chaining together the $B_i$ sets}\label{chaining}
Take $\alpha$ sufficiently small; in fact, what we need is that \[
\alpha < (d-1)^{-(C_4)^2/4}
\]using the $C_4$ from Lemma \ref{Bi4}. Let $G$ be the corresponding set given by Lemma \ref{allbutdelt} under the starting matching $\mc{M}$, where $G$ consists of all elements of BFS distance at most $M$ for some appropriate $M$, and let $B_{C_4}=X\setminus G$.

For all $i>C_4$, recursively define $B_i=B_i^{(4)}(B_{i-1})$.

\begin{lemma}\label{expsmallalli}
For any $d\ge 3$ and $c<c_d^*$, we have with high probability that for any matching $\mc{M}$, $|B_i|\le (d-1)^{-(i+1)^2/4}n$ for all $i\ge C_4$.
\end{lemma}
Note that, in particular, this implies that there is a $C'=\Theta(1)$ such that $B_i=\emptyset$ for all $i\ge C'\sqrt{\log n}$.

\begin{proof}
We will prove this by induction on $i$, noting that it is true for the base case $i=C_4$.  

\textbf{Case 1: }If $|B_{i-1}|<(d-1)^{-i^3}n$, then we note that $B_i^{(4)}(B_{i-1})\subseteq W_{-\le i}(B_{i-1})$, so \begin{align*}
|B_i|&\le|W_{-\le i}(B_{i-1})|\le \left(3d\log\left(\frac{n}{|B_{i-1}|}\right)\right)^i|B_{i-1}|\say{by Lemma \ref{nbrsofnbrs}}\\
&\le \left(3d\log\left(\frac{n}{(d-1)^{-i^3}n}\right)\right)^i(d-1)^{-i^3}n\\
&\le \left(3d\textit{}\left(i^3\log(d-1)\right)\right)^i(d-1)^{-i^3}n\\
&\le d^{-i^2/4}n
\end{align*}
as desired.

\textbf{Case 2: }Otherwise, we have that \[
(d-1)^{-i^3}n\le|B_{i-1}|\le(d-1)^{-i^2/4}n,
\]and thus we can apply Lemma \ref{Bi4} to say that \begin{align*}
|B_i|&\le(d-1)^{-0.6i}|B_{i-1}|\say{by Lemma \ref{Bi4}}\\
&\le(d-1)^{-0.6i}(d-1)^{-i^2/4}n\\
&\le (d-1)^{-i^2/4-i/2-1/4}n=(d-1)^{-(i+1)^2/4}n
\end{align*}as desired.
\end{proof}     

\begin{lemma}\label{iprocess}
Conditioned on starting at any vertex outside of $B_i$, the random walk has probability $\ge 0.99$ of being in $G$ (or having finished) in exactly $i(i+1)/2-C_4(C_4+1)/2$ steps.
\end{lemma}
\begin{proof}
For any $C_4<j\le i$, we see by the definition of $B_j$ that, if we are at a vertex outside of $B_j$, we have probability $\ge 1-\frac{2}{j^{1.5}}$ of being outside of $B_{j-1}$ after $j$ steps.

Iterating this, we see that the probability of being in $G$ or finished (that is, outside of $B_{C_4}$) after $\sum_{j=C_4+1}^ij=i(i+1)/2-C_4(C_4+1)/2$ steps is \[
\ge 
\prod_{j=C_4+1}^i\left(1-\frac{2}{j^{1.5}}\right)\ge 1-\sum_{j=C_4+1}^i\frac{2}{j^{1.5}}\ge 1-\sum_{j=C_4+1}^\infty\frac{2}{j^{1.5}}\ge 0.99
\]as desired (using $C_4\ge 160000$).
\end{proof}
\begin{lemma}\label{finish}
For any $i\ge C_4$, conditioned on starting at any vertex outside of $\\W_{-\le 10(M+1)(d-1)^M}(B_i)$, the random walk has probability $\ge 0.97$ of terminating (finishing in $U$) within $10(M+1)(d-1)^M+(i(i+1)/2-C_4(C_4+1)/2)$ steps.
\end{lemma}
\begin{proof}
For any random walk, denote by $x_t$ the location of the random walk after $t$ steps, with $x_0$ being the initial hash location. Recalling the constant $M$ and set $G$ from Lemma \ref{allbutdelt}, define an $M$-separated sequence in $G$ to be a list of steps $t_1,\dotsc,t_q$ such that $x_{t_1},\dotsc,x_{t_q}\in G$ and $t_{r+1}-t_r>M$ for every $1\le r<q$. That is, an $M$-separated sequence in $G$ is a list of times that the random walk is in $G$ where each time is at least $M$ steps after the previous. 

Given a random walk that goes until finishing at an empty slot, let $\zeta$ be the maximum length of an $M$-separated sequence in $G$. Then we have \[
\mb{P}(\zeta\ge s)\le(1-(d-1)^{-M})^{s-1},
\]as for every time $t$ where $x_t\in G$, we have that the random walk will be finished in at most $M$ more steps with probability at least $(d-1)^{-M}$. 

Furthermore, note that this inequality $\mb{P}(\zeta\ge s)\le(1-(d-1)^{-M})^{s-1}$ still holds conditioned on any given starting location for the random walk. 

Let $\mc{G}_1$ denote the event that $\zeta>5(d-1)^M$. Then $\mb{P}(\mc{G}_1)\le (1-(d-1)^{-M})^{5(d-1)^M}\le e^{-5}<0.01$. This is still true conditioned on starting at any vertex outside of $W_{-\le 10(M+1)(d-1)^M}(B_i)$.   

Now, for any $k\in\mb{N}$, let $E_k$ denote the event that step $k(M+1)+(i(i+1)/2-C_4(C_4+1)/2)$ of the random walk is in $G$ or finished. The fact that the walk starts outside of $W_{-\le 10(M+1)(d-1)^M}(B_i)$ means that for $k\le 10(d-1)^M$, step $k(M+1)$ of the random walk is not in $B_i$, and thus by Lemma \ref{iprocess}, there is probability at least 0.99 that step $k(M+1)+(i(i+1)/2-C_4(C_4+1)/2)$ of the random walk will either be in $G$ or finished.     

Therefore, $\mb{P}(E_k)\ge 0.99$ for all $k\le 10(d-1)^M$. (We do not claim that these events are independent.) Then the expected number of $k\in\{0,\dotsc,10(d-1)^M-1\}$ such that $E_k$ does not occur is at most $0.1(d-1)^M$.     

Let $\mc{G}_2$ be the event that there are at least $5(d-1)^M$ values $k\in\{0,\dotsc,10(d-1)^M-1\}$ such that $E_k$ does not occur. By Markov's inequality, $\mb{P}(\mc{G}_2)\le\frac{0.1(d-1)^M}{5(d-1)^M}=0.02$.   

Together, we get $\mb{P}(\mc{G}_1\cup\mc{G}_2)\le\mb{P}(\mc{G}_1)+\mb{P}(\mc{G}_2)\le 0.01+0.02=0.03$. 

Finally, we claim that if neither $\mc{G}_1$ nor $\mc{G}_2$ happen, then the random walk finishes within $10(M+1)(d-1)^M+(i(i+1)/2-C_4(C_4+1)/2)$ steps, which will complete the proof. 

If $\mc{G}_2$ does not happen, then there are more than $5(d-1)^M$ values of $k$ such that $E_k$ does occur. This means that there are more than $5(d-1)^M$ values of $k$ for which $k(M+1)+(i(i+1)/2-C_4(C_4+1)/2)$ is either in $G$ or finished. If all of those $k$ were in $G$ (and not finished), then that would produce a $M$-separated sequence in $G$ of length $\ge 5(d-1)^M$. However, $\mc{G}_1$ not occuring means that no such sequence exists. Therefore, there must be some $k\in\{0,\dotsc,10(d-1)^M-1\}$ such that step $k(M+1)+(i(i+1)/2-C_4(C_4+1)/2)$ is finished, completing the proof.
\end{proof}

\section{Improved Bounds on the Number of Walks to any Set}\label{improvedreaching}
We now seem to be very close to proving the theorem, as we have shown that there are sets $B_i$ such that $|B_i|$ declines exponentially with $i$, and there is $0.97$ probability in finishing in $O(i^2)$ steps when starting outside of $B_i$. However, we do still need to deal with what happens in the 0.03 probability case. To complete our proof of Theorem \ref{theorem}, we need to improve the bounds on Lemma \ref{nbrsofnbrs}.

Lemma \ref{nbrsofnbrs} showed that $|W_{-\le j}(S)|\le (O(\log(n/|S|)))^j|S|$. Intuitively, as an average slot has in expectation $d$ hashes to it, you should expect $|W_{-\le j}(S)|$ to grow like $d^j|S|$ for an average $S$. Our goal in this section is to get a result closer to $d^j|S|$ for an upper bound on $|W_{-\le j}(S)|$.

In Lemma \ref{nbrsofnbrs}, we implicitly did a union bound over all $S'\subseteq X$ (partitioning based on $|S'|$) at every single one of the $j$ steps. In order to get a stronger result, when going from $W_{-\le k}(S)$ to $W_{-\le (k+1)}(S)$, we can instead try to union bound over only the $S'\subseteq X$ such that $S'=W_{-\le 1}(S'')$ under some matching $\mc{M}$. The $S''$ have already been union bounded over at the previous step, that is, when going from $-(k-1)$ to $-k$.   

\begin{lemma}\label{reachingsmall}
For any $d\ge 3$ and $c<c_d^*$, we have with high probability that for any matching $\mc{M}$ any $S\subseteq X$ with $|S|<|X|/(10d)$, and for any $j\in\mb{N}$ (possibly depending on $n$), it is true that \[|W_{-\le j}(S)|\le \left(2d+\log(d)+2\right)^je^{5\left(\log^2\left(\log\left(\frac{n}{|S|}\right)\right)\right)}|S|.\]
\end{lemma}     

Note that this is trivially true for $j=\Omega(\log n)$, as the right hand side of the inequality would then be greater than $n$ for any non-trivial $S$, so we can assume, say, $j=O(\log^2(n))$.

\begin{proof}
First, imagine fixing some $S\subseteq X$ before any hashes are revealed. Then, we generate the hashes of the objects in $S$. We then perform a union bound over the $\le d^{|S|}$ choices for which out of $d$ slots each object in $S$ is occupying under $\mc{M}$. Next, we reveal which other hashes land in the slots that are occupied by $S$, thus determining $W_{-\le 1}(S)$.     

Then, we again continue iteratively, next generating the hashes of the objects in $W_{-\le 1}(S)$ and union bounding over the $\le d^{|W_{-\le 1}(S)|}$ choices of where they occupy.     

Note that for any fixed choice of the slots occupied by $W_{-\le k}(S)$, we have that $|W_{-\le (k+1)}(S)|-|W_{-\le k}(S)|$ is stochastically dominated by the binomial random variable $Bin(dn,|W_{-\le k}(S)|/n)$, and so $\mb{E}(|W_{-\le (k+1)}(S))|\le (d+1)|W_{-\le k}(S)|$. By standard Chernoff bounds, for any $\lambda>0$ and for any particular choice of the slots occupied by $W_{-\le k}(S)$, which we denote by the ``conditioning on $\mc{K}$'', or ``$\mid\mc{K}$'' symbol, we have
\[
\mb{P}\left(|W_{-\le (k+1)}(S)|\ge (1+\lambda)(d+1)|W_{-\le k}(S)|\mid\mc{K}\right)\le e^{-\lambda^2(d+1)(|W_{-\le k}(S)|)/(\lambda+2)}.
\]

Intuitively, this means that $|W_{-\le j}(S)|$ should on average be upper bounded by $(d+1)^j|S|$. We will use these Chernoff bounds to get a result that holds even in our worst case. Let \[\lambda_k=\frac{\log(d)}{d+1}+\frac{12\log^2(en/|S|)}{(d+1)^k}+1.\]
We will induct on $k$ to bound $\mb{P}\left(|W_{-\le (k+1)}(S)|\ge \left(\prod_{\ell=0}^{k}(1+\lambda_\ell)\right)(d+1)^{k+1}|S|\right)$.     

We claim that \begin{equation}
\frac{\lambda_k^2}{\lambda_k+2}\ge\frac{\log(d)}{d+1}+\frac{6\log^2(en/|S|)}{(d+1)^k}\label{C1C2Result}
\end{equation}
To show Inequality \eqref{C1C2Result}, we can set $C_1=\frac{\log(d)}{d+1}\in(0,1)$ and $C_2=\frac{12\log^2(en/|S|)}{(d+1)^k}\ge 0$. Then we see that $0.5C_2^2+0.5C_1C_2+0.5C_2\ge C_1-1$ is true whenever $C_1\in(0,1)$ and $C_2>0$, as the left side will then be positive while the right side is negative. This leads us to show that\begin{align*}
0.5C_2^2+0.5C_1C_2+0.5C_2&\ge C_1-1\\\implies
C_1^2+C_2^2+2C_1C_2+2C_1+2C_2+1&\ge C_1^2+0.5C_2^2+1.5C_1C_2+3C_1+1.5C_2\\\implies 
\frac{(C_1+C_2+1)^2}{C_1+C_2+3}&\ge C_1+0.5C_2,
\end{align*}giving Inequality \eqref{C1C2Result}.

So, assume (inductively) that $|W_{-\le k}(S)|\leq \left(\prod_{\ell=0}^{k-1}(1+\lambda_\ell)\right)(d+1)^k|S|$. We can also assume that $|W_{-\le k}(S)|\ge \left(\prod_{\ell=0}^k(1+\lambda_\ell)\right)(d+1)^k|S|/(3d\log(n/|S|))$, as otherwise we are finished by Lemma \ref{Upperboundnbrs}. We have for any particular choice of the slots occupied by $W_{-\le k}(S)$ that
\begin{align*}
\mb{P}\left(|W_{-\le (k+1)}(S)|\ge \left(\prod_{\ell=0}^{k}(1+\lambda_\ell)\right)(d+1)^{k+1}|S|\,\bigg|\,\mc{K}\right)&\le
\mb{P}\left(|W_{-\le (k+1)}(S)|\ge (1+\lambda_k)(d+1)|W_{-\le k}(S)|\mid\mc{K}\right)
\\&\say{by the inductive hypothesis}
\\&\le e^{-\lambda_k^2(d+1)(|W_{-\le k}(S)|)/(\lambda_k+2)}
\\&\say{by standard Chernoff bounds}
\\&\le e^{-\left(\log(d)|W_{-\le k}(S)|+6(d+1)|W_{-\le k}(S)|\log^2(en/|S|)/(d+1)^k\right)}
\\&\say{by Inequality \eqref{C1C2Result}}
\\&\le e^{-\left(\log(d)|W_{-\le k}(S)|+2\left(\prod_{\ell=0}^k(1+\lambda_\ell)\right)|S|\log(en/|S|)\right)}
\\&\say{\hspace{-.5cm}Using the lower bound on $|W_{-\le k}(S)|$ above}
\\&\le e^{-(\log(d)|W_{-\le k}(S)|+2|S|\log(en/|S|))}
\\&\le d^{-|W_{-\le k}(S)|}\left(\frac{en}{|S|}\right)^{-2|S|}.
\end{align*}
Now, we union bound over the $\le d^{|W_{-\le k}(S)|}$ choices we faced to choose which slots the objects in $W_{-\le k}(S)$ were matched to. Therefore,\[
\mb{P}\left(|W_{-\le (k+1)}(S)|\ge \left(\prod_{\ell=0}^{k}(1+\lambda_\ell)\right)(d+1)^{k+1}|S|\right)\le\left(\frac{en}{|S|}\right)^{-2|S|}.
\]
Summing over all $1\le k\le j$ gives us that \[
\mb{P}\left(|W_{-\le j}(S)|\ge \left(\prod_{k=0}^{j-1}(1+\lambda_k)\right)(d+1)^j|S|\right)\le j\left(\frac{en}{|S|}\right)^{-2|S|}.
\]
Now, we union bound over the $\binom{n}{s}\le\left(\frac{en}{s}\right)^s$ choices for $S$ with $|S|=s$ to say that \begin{align*}
\mb{P}\left(\exists~S\subseteq X\text{ s.t. }|W_{-\le j}(S)|\ge \left(\prod_{\ell=0}^{j-1}(1+\lambda_\ell)\right)(d+1)^j|S|\right)&\le \sum_{s=1}^n j\left(\frac{en}{s}\right)^{-s}\\&\le\frac{j}{n}\sum_{s=1}^nse^{-s}\left(\frac sn\right)^{s-1}\le\frac{j}{n}\sum_{s=1}^nse^{-s}\\&\le\frac jn\frac{e}{(e-1)^2}\le o(1)
\end{align*}
as long as $j=o(n)$, which is true as it is $O(\log^2(n))$. 

Therefore, we have proven that with high probability, for every $S\subseteq X$ and $0\le j\le\log^2(n)$, we have that \[
|W_{-\le j}(S)|\le\left(\prod_{k=0}^{j-1}(1+\lambda_k)\right)(d+1)^j|S|.
\]
Now, what remains is to upper bound the product $\prod_{\ell=0}^{j-1}(1+\lambda_k)$.
\begin{align*}
&\left(\prod_{k=0}^{j-1}(1+\lambda_k)\right)(d+1)^j|S|\\
&=\left(\prod_{k=0}^{j-1}\left(1+\frac{\log(d)}{d+1}+\frac{12\log^2(en/|S|)}{(d+1)^k}+1\right)\right)(d+1)^j|S|
\\&=\left(\prod_{k=0}^{j-1}\left(1+\frac{12\log^2(en/|S|)}{(2+\log(d)/(d+1))(d+1)^k}\right)\right)\left (2d+\log(d)+2\right)^j|S|
\\&\le\left(\prod_{k=0}^\infty\left(1+\frac{\log^2(en/|S|)}{(d+1)^{k-2}}\right)\right)\left (2d+\log(d)+2\right)^j|S|
\\&\le\left(\prod_{k=0}^\infty\left(1+\frac{\log^2(en/|S|)}{(d+1)^{k-2}}\right)\right)\left (2d+\log(d)+2\right)^j|S|
\\&\le\left(\prod_{k=0}^{2\log_3\left(\log\left(\frac{en}{|S|}\right)\right)+1}\left(1+\frac{\log^2(en/|S|)}{(d+1)^{k-2}}\right)\right)\\&\hspace{3cm}\left(\prod_{k=2\log_3\left(\log\left(\frac{en}{|S|}\right)\right)+2}^\infty\left(1+\frac{\log^2(en/|S|)}{(d+1)^{k-2}}\right)\right)\left (2d+\log(d)+2\right)^j|S|
\\&\le\left(\prod_{k=0}^{2\log_3\left(\log\left(\frac{en}{|S|}\right)\right)+1}\left(1+\frac{\log^2(en/|S|)}{(d+1)^{k-2}}\right)\right)\left(\prod_{k=0}^\infty\left(1+\frac{1}{3^k}\right)\right)\left (2d+\log(d)+2\right)^j|S|
\\&\le  4\left(\prod_{k=0}^{2\log_3\left(\log\left(\frac{en}{|S|}\right)\right)+1}\left(1+\frac{\log^2(en/|S|)}{(d+1)^{k-2}}\right)\right)\left (2d+\log(d)+2\right)^j|S|
\\&\le {4}\left(\prod_{k=0}^{2\log_3\left(\log\left(\frac{en}{|S|}\right)\right)+1}\left(1+(d+1)^2\log^2(en/|S|)\right)\right)\left (2d+\log(d)+2\right)^j|S|
\\&\le 12\left(\left(1+(d+1)^2\log^2(en/|S|)\right)^{2\log_3(\log(en/|S|))}\right)\left (2d+\log(d)+2\right)^j|S|
\\&\le {12}e^{2\left(\log\left(1+(d+1)^2\log^2(en/|S|)\right)\log_3(\log(en/|S|))\right)}\left (2d+\log(d)+2\right)^j|S|
\\&\le e^{5\left(\left(\log^2\left(\log(en/|S|)\right)\right)\right)}\left (2d+\log(d)+2\right)^j|S|
\end{align*}
for sufficiently large $n$, as desired.
\end{proof}

\section{Proof of Theorem \ref{theorem}}\label{finalproofsection}

\begin{theorem}
For any $d\ge 3$ and $c<c_d^*$, we have with high probability that the expected insertion time is $O(1)$.
\end{theorem}
\begin{proof}
Essentially, the idea here is that if we fail to finish in $\le i^2$ on the run starting outside $B_i$ (which happens with probability at most 0.03), then we are probably still outside $B_{3i}$, so we try again with a run starting there, then $B_{9i}$, and so on.     

Note that $|B_{3^ki}|\le (d-1)^{-9^ki^2/4}n$ by Lemma \ref{expsmallalli}. Furthermore, if the walk is currently outside of $B_{3^ki}$, then Lemma \ref{finish} says that we have probability at least 0.97 of finishing in $10(M+1)(d-1)^M+(3^ki(3^ki+1)/2-C_4(C_4+1)/2)$ steps. For sufficiently large $i$, we have that $10(M+1)(d-1)^M+(3^ki(3^ki+1)/2-C_4(C_4+1)/2)\le (0.51)9^ki^2$.     

Formally, for all $i\ge C_5$ for a sufficiently large constant $C_5$, let $\mc{E}_i$ be the event that the starting hash of the object we are inserting is outside of $W_{-\le 9^ki^2/15}(B_{3^ki})$ for every $k\in\mb{Z}_{\ge 0}$. Note that if $C_5$ is sufficiently large (specifically, if $C_5^2/15\ge 10(M+1)(d-1)^M$), then the $k=0$ case of this hypothesis includes being outside of $W_{-\le (10(M+1)(d-1)^M)}(B_i)$, satisfying the hypothesis of Lemma \ref{finish}.\\

We now bound $\sum_k|W_{-\le 9^ki^2/15}(B_{3^ki})|$. For all $i\ge C_5$, we have
\begin{align*}
|W_{-\le 9^ki^2/15}(B_{3^ki})|&\le O\left((2d+\log(d)+2)^{9^ki^2/15}e^{5\log^2\left(\log\left(\frac{n}{|B_{3^ki}|}\right)\right)}|B_{3^ki}|\right)
\\&\say{by Lemma \ref{reachingsmall}}
\\&\le O\left((2d+\log(d)+2)^{9^ki^2/15}e^{5\log^2\left((3^ki)^2\log(d-1)/2\right)}(d-1)^{-(3^ki)^2/4}n\right)
\\&\say{by Lemma \ref{expsmallalli}}
\\&\le O\left((2d+\log(d)+2)^{9^ki^2/15}e^{\left(\log^3(3^ki)\right)}(d-1)^{-9^ki^2/4}n\right)
\\&\le O\left(e^{\left(\log^3(3^ki)+9^ki^2\left(\log(2d+\log(d)+2)/15-\log(d-1)/4\right)\right)}n\right)
\\&\le O\left(e^{-9^ki^2/50}n\right)\say{for $i\ge C_5$ and $d\ge 3$}
\end{align*}

Then \begin{align*}
\left|\bigcup_{k=0}^\infty W_{-\le 9^ki^2/15}(B_{3^ki})\right|&\le\sum_{k=0}^\infty O\left(e^{-9^ki^2/50}n\right)
\\&\le O(e^{-i^2/50}n)
\end{align*}
In particular, this means that $\mc{E}_i$ happens with probability at least $1-O(e^{-i^2/50})$.     

Conditioned on starting on any specific vertex under which $\mc{E}_i$ happens, we claim that the expected run-time is $O(i^2)$. By Lemma \ref{finish}, since we started outside of $W_{-\le (10(M+1)(d-1)^M)}(B_i)$, there is a $\ge 0.97$ probability of finishing in $\le 0.51i^2$ steps (again using $i\ge C_5$). If we do not finish after $0.51i^2$ steps (the $\le0.03$ event occurs), then because we started outside of $W_{-\le 9i^2/15}(B_{3i})$ and thus outside of $W_{-\le (0.51i^2+10(M+1)(d-1)^M)}(B_{3i})$, we are still outside of $W_{-\le (10(M+1)(d-1)^M)}(B_{3i})$. Then by Lemma \ref{finish}, there is now a $0.97$ probability of finishing in $0.51(3i)^2$ more steps.     

In general, after $k$ iterations we have taken 
\[
\sum_{q=0}^k0.51(3^qi)^2\le (0.58)9^ki^2<(9^{k+1})i^2/15-10(M+1)(d-1)^M
\] 
steps, so we are still outside of $W_{-\le (10(M+1)(d-1)^M)}(B_{3^{k+1}i})$. The chance of reaching the $k$-th stage without finishing is at most $0.03^{k}$, and the number of steps taken through the $k$-th stage is $(0.58)9^ki^2$. Therefore, the total expected number of steps taken is at most 
\[
\sum_{k=0}^\infty (0.03^k)((0.58)9^ki^2)\le i^2
\]     

So, conditioned on starting at a given vertex under which $\mc{E}_i$ happens, the expected run time of the random walk is at most $i^2$.     

For any $i\ge C_5$, let $\mc{F}_i$ be the event that $\mc{E}_i$ happens but $\mc{E}_j$ does not happen for every $C_5\le j<i$. This is a partition of where our starting hash lands. Let $T$ be the run time of our random walk. Then by the law of total probability,
\begin{align*}
\mb{E}(T)&=\sum_{i=C_5}^{\infty}\mb{E}(T|\mc{F}_i)\mb{P}(\mc{F}_i)
\\&\le(C_5)^2+\sum_{i=C_5+1}^{\infty}\mb{E}(T|\mc{F}_i)\mb{P}(\mc{F}_i)\say{as we start at a vertex where $\mc{E}_{C_5}$ happens}
\\&\le O(1)+\sum_{i=C_5+1}^{\infty}(i^2)\mb{P}(\mc{F}_i)\say{as we start at a vertex where $\mc{E}_i$ happens}
\\&\le O(1)+\sum_{i=C_5+1}^{\infty}(i^2)(O(e^{-(i-1)^2/50}))\say{as $\mc{E}_{i-1}$ does not happen}
\\&\le O(1)
\end{align*}as desired.
\end{proof}

\section{Tail Bounds}\label{tailsection}
To prove that the expected insertion time is $O(1)$, it would have sufficed, for instance, to prove that the probability of the random walk taking more than $\ell$ insertion steps is $O(\ell^{-1.01})$. In this section, we will give a stronger bound on the probability of the random walk taking more than $\ell$ insertion steps. The following corollary shows that the probability of taking more than $\ell$ insertion steps is $O(\ell^{-1000})$, where the value 1000 can be made arbitrarily large.
\begin{corollary}\label{tailbounds}
Let $C_6>1$. There exists a constant $C_7=C_7(C_6,c,d)=\Theta(1)$ such that for all $\ell\in\mb{N}$, the probability that the random walk takes more than $\ell$ steps is at most $C_7\ell^{-C_6}$ as $n\rightarrow\infty$.    
\end{corollary}
\begin{proof}

First, note that for any $\epsilon_1>0$, the value 0.99 in Lemma \ref{iprocess} can be replaced with $1-\epsilon_1$ by requiring $C_4$ to be large enough such that $\sum_{j=C_4}^\infty\frac{2}{j^{1.5}}\le\epsilon_1$.  

Correspondingly, the 0.97 in Lemma \ref{finish} can be replaced with $1-\epsilon_2$ for any $\epsilon_2>0$ as well. This is because we can take $\epsilon_1=\epsilon_2/3$, replace $10(d-1)^M$ with $2\log(3/\epsilon_2)(d-1)^M$ and $5(d-1)^M$ with $\log(3/\epsilon_2)(d-1)^M$.

Then, letting $\mc{G}_1$ in Lemma \ref{finish} denote the event that $\zeta\ge \log(3/\epsilon_2)(d-1)^M$, we get $\mb{P}(\mc{G}_1)\le (1-(d-1)^{-M})^{\log(3/\epsilon_2)(d-1)^M}\le e^{-\log(3/\epsilon_2)}=\epsilon_2/3$.

Similarly, we use the same definition of $E_k$ and let $\mc{G}_2$ in Lemma \ref{finish} denote the event that there are at least $\log(3/\epsilon_2)(d-1)^M$ values $k\in\{0,\dotsc,2\log(3/\epsilon_2)(d-1)^M-1\}$ such that $E_k$ does not occur. we get $\mb{P}(\mc{G}_2)\le \frac{\epsilon_1(2\log(3/\epsilon_2)(d-1)^M)}{\log(3/\epsilon_2)(d-1)^M}=2\epsilon_1=2\epsilon_2/3$.

This gives a total failure probability of at most $\epsilon_2$.  

Now, when considering the probability that the random walk takes at least $\ell$ steps, we partition on whether $\mc{E}_{\log(\ell)}$ happens, where the definition of $\mc{E}_i$ is taken from Section \ref{finalproofsection}.  

The probability that $\mc{E}_{\log(\ell)}$ does not happen is $O(e^{-\log^2(\ell)/50})=O(\ell^{-\log(\ell)/50})\le O(\ell^{-C_6})$ for sufficiently large $\ell$.  

If $\mc{E}_{\log(\ell)}$ does happen, then as in Section \ref{finalproofsection}, the probability of not being finished after $k$ iterations is at most $(\epsilon_2)^k$, and the total number of steps taken up to iteration $k$ is $\le (0.58)9^k(\log(\ell))^2$. Taking $k$ to be $\log_9(\sqrt\ell)$, we see that the probability of having taken more than \[
(0.58)9^{\log_9(\sqrt\ell)}(\log(\ell))^2=0.58\sqrt\ell(\log(\ell))^2<\ell
\]steps is at most\[
(\epsilon_2)^{\log_9(\sqrt\ell)}=\ell^{\log_9(\sqrt{\epsilon_2})}\le\ell^{-C_6}
\]as long as we have made $\epsilon_2$ sufficiently small in terms of $C_6$. (We also implicitly assume that $\ell$ is sufficiently large in terms of $C_6$, $d$, and $\epsilon_2$ by increasing $C_7$ accordingly.)
\end{proof}

This shows that the tail bounds on the random walk decline faster than any polynomial. This does not show that these tail bounds are exponential: for instance, we have not excluded the possibility  that the probability of the random walk taking $\ell$ steps is $\Theta(e^{-(\log^2(\ell))})$. We believe that the true tail bounds should be exponentially decreasing (for some base of the exponent):
\begin{conjecture}\label{conjecturedtail}
There exists constants $C_8>0$ and $C_9>1$ such that for all $\ell\in\mb{N}$, the probability of the random walk taking at least $\ell$ steps is at most $C_8((C_9)^{-\ell})$.
\end{conjecture}

\section{$O(1)$ Expected Insertion for BFS Insertion}\label{BFSSection}

An insertion algorithm that differs significantly from random walk insertion is BFS insertion. Recall that BFS insertion refers to the insertion algorithm where we compute the shortest augmenting walk and reassign objects along that. There are $d(d-1)^{i-1}$ possibilities for walks of length $i$. We will now note that $O(1)$ expected time for BFS insertion comes as a corollary of Lemmas \ref{allbutdelt} and \ref{FPSneighbors}.
\begin{corollary}\label{BFS}
Let $d\ge 3$ and $c<c_d^*$. With high probability, BFS Insertion takes $O(1)$ expected time.
\end{corollary}
\begin{proof}
For all $i\in\mb{N}$, let $D_i$ be the set of all elements at BFS distance of at least $i$ from $U$, the set of unoccupied slots. Note that every slot in $N(D_{i+1})$ must be occupied by an object in $D_i$, so $|N(D_{i+1})|\le |D_i|$.

By Lemma \ref{FPSneighbors}, we note that there is a constant $\alpha=\Theta(1)$ such that if $1\le |S|\le\alpha |X|$, then $|N(S)|\ge(d-1.5)|S|$, as making $|X|/|S|$ a sufficiently large constant makes $p_{|S|}<0.5$. Apply Lemma \ref{allbutdelt} with this $\alpha$, and let $M$ be the constant that results. The previous paragraph then implies that for every $i\ge M$, we have $|D_{i+1}|(d-1.5)\le|N(D_{i+1})|\le |D_i|$. Applying this iteratively, we get that $|D_{i+M}|\le (d-1.5)^{-i}|D_M|\le (d-1.5)^{-i}n$ for every $i\ge M$. Then there exists a $C=\Theta(1)$ (in particular, $C=(d-1.5)^M$) such that for every $i\in\mb{N}$, $|D_i|\le C(d-1.5)^{-i}n$.     

The run time of BFS insertion on an object $x$ that is at BFS distance $i$ from $U$ can be bounded by $O((d-1)^i)$, as noted in \cite{FPSS03}. Then using a similar argument to \cite{FPSS03}, we find that \begin{align*}
\mb{E}(\text{BFS Insertion Time})&=O\left(\sum_{i\in\mb{N}}(d-1)^{i+1}\mb{P}(x\text{ at BFS distance }>i)\right)\\&=O\left(\sum_{i\in\mb{N}}(d-1)^{i+1}\mb{P}(h_j(x)\text{ occupied by some }x'\in D_i~\forall~1\le j\le d)\right)\\&=O\left(\sum_{i\in\mb{N}}(d-1)^{i+1}\left(\frac{|D_i|}{n}\right)^d\right)\\&=O\left(\sum_{i\in\mb{N}}(d-1)^{i+1}(d-1.5)^{-id}\right)=O(1)
\end{align*}as $(d-1)(d-1.5)^{-d}<1$ for all $d\ge 3$.
\end{proof}         \section{Future Work}\label{futureworksection}  

One line of improvement would be to improve the tail bounds on the number of steps in the random walk beyond what was proven in Corollary \ref{tailbounds}. Proving (or disproving) Conjecture \ref{conjecturedtail} would be a good goal, though weaker improvements would also be worthwhile.

It would also be interesting to give a stronger bound on the $o(1)$ term in our ``with high probability'' statements. A careful analysis of our and previous works (\cite{FP12,FPS13}) shows that this probability (originating from Lemmas \ref{allbutdelt}, \ref{Upperboundnbrs}, \ref{fewcycles}, \ref{FPSneighbors}, and \ref{reachingsmall}) could currently be taken to be $O(n^{-\beta})$ for some small $\beta=\Theta(1)$.

Now that we have an insertion time independent of $n$, another avenue for future study is to optimize the insertion time in terms of $d$, $c$, and absolute constants. Subsequent to the initial version of this paper, Kuszmaul and Mitzenmacher have done work along this line \cite{KM25}.     

It has been shown under some previous models of cuckoo hashing that the assumption of uniformly random hash functions can be relaxed to families of efficiently computable hash functions while retaining the theoretical insertion time guarantees \cite{indephash,explicithash}. As our proof relies on similar ``expansion-like'' properties of the bipartite graph to previous work, we believe that Theorem \ref{theorem} should still hold under practically computable hash families.     

A different model for generalizing cuckoo hashing, proposed in 2007, gives a capacity greater than one to each hash table slot (element of $Y$), instead of (or in addition to) additional hash functions \cite{multistorage}. The load thresholds for this model are known for both two hashes \cite{binsthreshold1,binsthreshold2} and $d\ge 3$ hashes \cite{FKPloadthresholds}. As in our model, $O(1)$ expected time for random walk insertion has been shown for some values below the load threshold \cite{FriezePetti,Walzer}, but it remains open for any capacities greater than one to prove $O(1)$ insertion up to the load thresholds.    

In general, it would be nice to extend our random walk insertion time guarantees to other modifications of cuckoo hashing, such as those schemes that increase the probability of a valid matching \cite{origstash,newstash,betterdary}.
\subsection*{Acknowledgment}We thank Stefan Walzer and the anonymous referees for their helpful comments and discovering issues with a previous version.
\bibliographystyle{alpha}
\bibliography{Cuckoo}
\end{document}